%% file: main.tex
\documentclass[conference]{IEEEtran}

\usepackage{cite}

\newcommand{\equalcontrib}{\textsuperscript{*}}
\newcommand{\affiliationnote}{\textsuperscript{$\dagger$}}

\author{
    \IEEEauthorblockN{Mengqian Zhang\equalcontrib\affiliationnote}
	\IEEEauthorblockA{Carnegie Mellon University, IC3\\
		mengqianzhang@cmu.edu}
	\and
	\IEEEauthorblockN{Sen Yang\equalcontrib}
	\IEEEauthorblockA{Yale University, IC3\\
		sen.yang@yale.edu}
	\and
	\IEEEauthorblockN{Kartik Nayak}
	\IEEEauthorblockA{Duke University\\
		kartik@cs.duke.edu}
	\and
	\IEEEauthorblockN{Fan Zhang}
	\IEEEauthorblockA{Yale University, IC3\\
		f.zhang@yale.edu}
}

\usepackage{tikz}
\usepackage{amsmath}
\usepackage{booktabs} %
\usepackage[hidelinks]{hyperref}
\usepackage[camera]{dtrt}

\usepackage[dvipsnames]{xcolor}
\usepackage{makecell}
\usepackage{mathtools}
\usepackage{url}
\usepackage{bbm}

\usepackage[vlined,ruled,linesnumbered]{algorithm2e} 
\SetKwRepeat{Do}{do}{while}

\usepackage{multirow}
\usepackage{verbatim} %
\usepackage{amsthm}   %
\usepackage{enumitem} %
\usepackage{thmtools} %

\newtheorem{theorem}{Theorem}
\newtheorem{corollary}{Corollary}

\newtheorem{proposition}{Proposition}
\newtheorem{definition}{Definition}
\newtheorem{claim}{Claim}

\newtheorem{remark}{Remark}
\newtheorem{property}{Property}

\newtheorem{observation}{Observation}
\newtheorem{example}{Example}

\usepackage{tcolorbox}
\usepackage{threeparttable}
\usepackage{graphicx}
\usepackage{subfigure}
\usepackage{tikz}
\usetikzlibrary{automata, positioning, arrows}
\tikzset{
box/.style ={
rectangle, %
rounded corners =1pt, %
minimum width =50pt, %
minimum height =20pt, %
inner sep=5pt, %
draw=black, %
align = center
}}

\newcommand{\set}[1]{\left\{ #1 \right\}}

\newcommand{\numCG}{\ensuremath{\#CG}}
\newcommand{\sizeCG}{\ensuremath{|CG|}}

\def\Snospace~{\S{}}

\ifdefined\corollaryautorefname
\renewcommand{\corollaryautorefname}{Cor.}
\else
\newcommand{\corollaryautorefname}{Cor.}
\fi
\ifdefined\observationautorefname
\renewcommand{\observationautorefname}{Obs.}
\else
\newcommand{\observationautorefname}{Obs.}
\fi

\newcommand{\betterratio}{53.2\%\xspace}
\newcommand{\name}{\textsf{Boost+}\xspace}
\newcommand{\mechanism}{\ensuremath{\mathcal{M}_{\name}}\xspace}

\makeatletter
\@ifpackagewith{dtrt}{draft}{
    \setlength{\marginparwidth}{1.5cm}
    \usepackage[colorinlistoftodos,prependcaption,textsize=scriptsize,textwidth=\marginparwidth]{todonotes}
    \newcommand{\fanz}[1]{{\todo[fancyline,color=red!40,caption={FZ}]{#1}}}
}{
\newcommand{\fanz}[1]{}
}
\makeatother

\makeatletter
\@ifundefined{parhead}
  {\newcommand{\parhead}[1]{\par\noindent\textbf{#1.}}}
  {\renewcommand{\parhead}[1]{\par\noindent\textbf{#1.}}}
\makeatother

\begin{document}
\title{Boost+: Equitable, Incentive-Compatible Block Building}

\maketitle

\begingroup
\renewcommand{\thefootnote}{*}
\footnotetext{Both authors contributed equally to this work.}
\renewcommand{\thefootnote}{\textdagger}
\footnotetext{This work was conducted while the author was at Yale University.}
\endgroup

\begin{abstract}
Block space on the blockchain is scarce and must be allocated efficiently through block building. However, Ethereum's current block-building ecosystem, MEV-Boost, has become highly centralized due to \emph{integration}, which distorts competition, reduces blockspace efficiency, and obscures MEV flow transparency. To improve equitability and economic efficiency in block building, we propose \name, a system that pools submitted order flows inside an attested execution environment and gives all participating block-building algorithms equal access to them.

The core of \name is the mechanism \mechanism, built around a baseline block-building algorithm.
\mechanism aligns incentives for both searchers (intermediaries that generate or route transactions) and builders: 
Truthful bidding is a dominant strategy for all builders. For searchers, truthful reporting is dominant whenever the baseline algorithm dominates competing builders, and it remains dominant for all conflict-free transactions, even when builders may win. 
We further show that even if a searcher can technically integrate with a builder, non-integration combined with truthful bidding still dominates any deviation for conflict-free transactions. 
We also implement a concrete baseline algorithm informed by empirical analysis of real-world transactions and evaluate its efficacy using historical data.
\end{abstract}

\IEEEpeerreviewmaketitle

\input{sections/1-intro}
\input{sections/2-background}

\input{sections/3-model}
\input{sections/4-mechanism}

\input{sections/5-implementation}

\input{sections/6-discussion}

\input{sections/7-related-works}

\input{sections/8-conclusions}

\section*{Acknowledgements}
We thank the anonymous reviewers for their constructive feedback.
We are grateful to the Flashbots Data team for their support of our empirical analysis.
We also thank Bruno Mazorra, Burak Öz, and Christoph Schlegel for helpful discussions, and Matt Weinberg and Aviv Zohar for discussions on graph-based block-building optimization.
This work was supported in part by NSF Award No. 2237814.
Sen Yang and Mengqian Zhang are supported in part by research funding from the Ethereum Foundation and the Sui Foundation. Sen Yang was also supported in part through a subaward under NSF Grant No. 2207214.

\input{sections/ethics.tex}

\bibliographystyle{IEEEtran} 
\bibliography{ref}
\appendices

\input{sections/appendix.tex}

\end{document}

%% file: sections/1-intro.tex
\section{Introduction}

Decentralized blockchains have the potential to enable novel applications with unprecedented security properties, such as stablecoins and Decentralized Finance~\cite{bbc2017cryptokitties,lazzarin2020five,chainalysis2021NFT,heimbach2025early}. 
However, security and decentralization requirements imply that the computational capacity of a decentralized blockchain is inherently {\em scarce}. 
To utilize scarce block space effectively, block proposers select and order transactions to form blocks whose execution yields the most profit. 

In smart contract blockchains such as Ethereum, the presence of maximal extractable value (MEV) makes block building computationally complex~\cite{daian2020flash}, which is believed to threaten decentralization by creating a high entry barrier that only well-resourced proposers can overcome.
In 2022, Ethereum introduced MEV-Boost~\cite{flashbots2025mevboost}, a system that aims to lower the entry barrier for proposers, thereby promoting diverse participation and enhancing decentralization.

However, this system also introduces undesired dynamics, driven by the formation of {\em integrations}~\cite{wahrstatter2024blockchain,eip7805,yang2024decentralization}.
Integration refers to the practice in which searchers\footnote{Searchers are intermediaries that generate or route transactions. We will elaborate on the definition in~\autoref{sec:background}.} send their profitable transactions exclusively to a single builder rather than broadcasting them broadly. This exclusive access gives the integrating builder an unfair advantage in block building and enables them to win the MEV-Boost auction with higher profit, thereby increasing their overall revenue. 

Integration contributes to the centralization of the builder market~\cite{yang2022sok,wahrstatter2023time,yang2024decentralization,oz2024wins}, which weakens security~\cite{wadhwa2025aucil} and competition (e.g., as of August 2025, two builders produce the majority of blocks~\cite{relayscan}).
Beyond centralization, integration also leads to inefficient use of scarce block space. Different integrated builders often receive valuable transactions from different searchers that do not conflict with each other. However, since only one builder's block is chosen, many of these compatible transactions are excluded even though they could have been included together in a single block.
Moreover, integration is typically paired with opaque off-chain payments, which obscure how money and MEV actually flow among parties. Such opacity creates information asymmetries, makes it harder for the ecosystem to form a clear picture of MEV extraction and revenue sharing, and complicates any future regulatory or compliance considerations.
Finally, the advantage of integration indicates that currently, the ecosystem tends to reward private partnerships over technical improvements (e.g., designing better algorithms for constructing more valuable blocks). 

These issues motivate a natural question: {\em Can we eliminate the incentive towards integration through mechanism and system design?}

\parhead{Our approach}
Achieving incentive compatibility for searchers and builders simultaneously proves challenging in the existing paradigm set by MEV-Boost, because the system is unable to enforce any rules regarding the interaction between searchers and builders. As illustrated in~\autoref{fig:overall}, the current design does not consider searchers as participants.

For the same reason, proposed changes within the MEV-Boost framework, such as BuilderNet~\cite{buildernet2024introducing}, fall short of providing strong incentive compatibility guarantees. %

Our solution, \name, departs from this paradigm and explicitly governs the interaction between searchers and builders.
In \name, builders submit block-building algorithms and searchers provide their transactions, also known as their order flows, to match the established terminology. \name then executes the submitted algorithms against submitted transactions to construct the best possible block and determines the value allocation among parties, following a mechanism with formal incentive compatibility guarantees.
The entire workflow runs in a TEE, as with BuilderNet, to protect the confidentiality of order flows and the integrity of the entire process.

\begin{figure}
    \centering
    \includegraphics[width=\columnwidth]{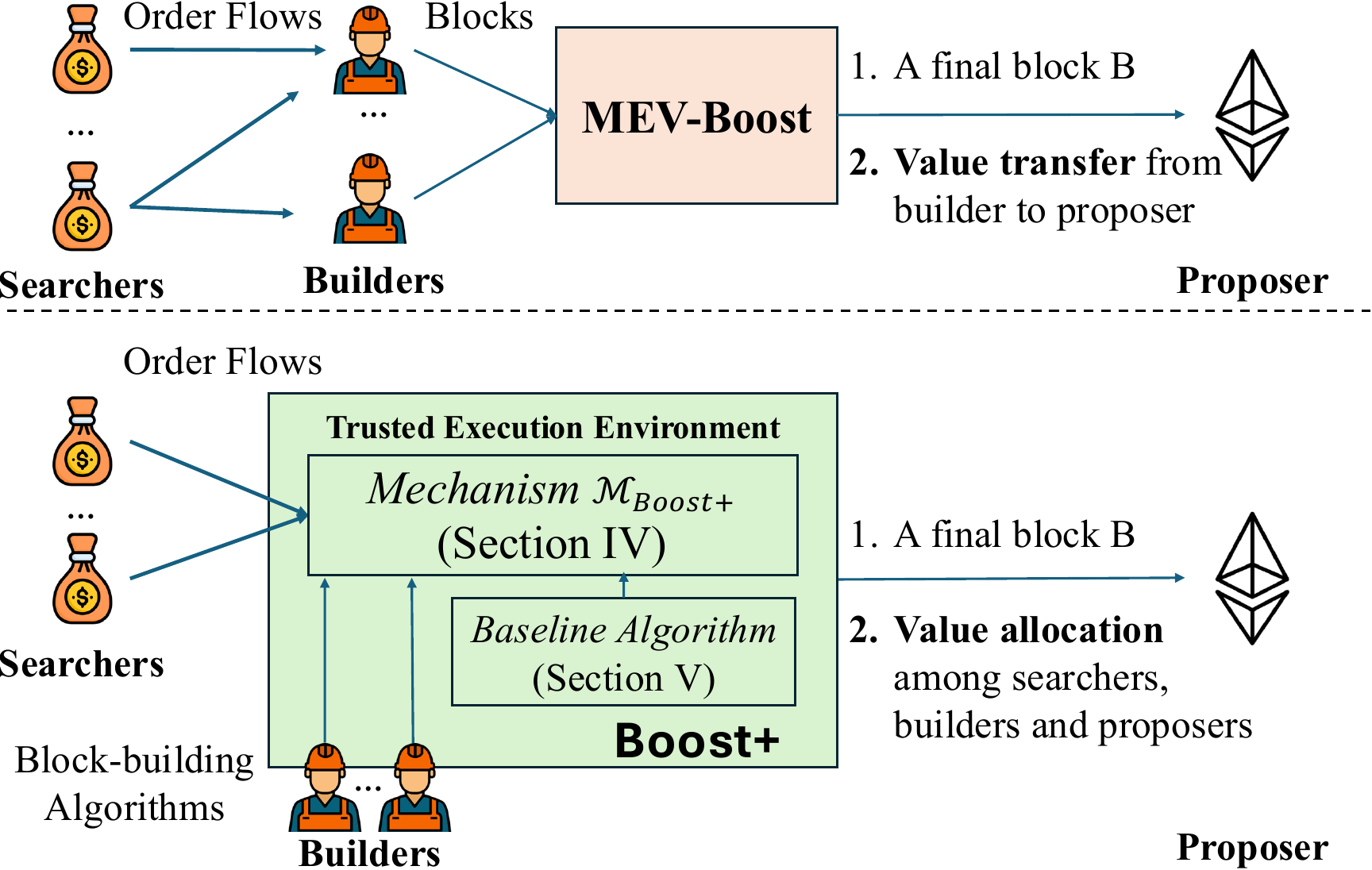}
    \caption{High-level comparison between \name and MEV-Boost.}
    \label{fig:overall}
\end{figure}

\subsection{Overview of our solution}
\parhead{The mechanism}
In the framework depicted above, the key component is a mechanism, denoted \mechanism, that determines the final block and associated money transfers among parties. 
The main challenge is to ensure truthfulness, meaning that for both searchers and builders, reporting their true valuations is a dominant strategy. \done%
To this end, we propose a mechanism based on the Vickrey--Clarke--Groves (VCG) framework~\cite{vickrey1961counterspeculation,clarke1971multipart,groves1973incentives}, the heart of which is a computationally efficient algorithm satisfying two key properties: (i) the set of candidate blocks it explores is fixed independently of bids, and (ii) within this fixed set, it can efficiently identify both the block maximizing the total bid and, for each transaction, the counterfactual block that would have been chosen had that particular transaction's bid been ignored from the maximization process. We refer to this algorithm as the \emph{baseline block-building algorithm}, emphasizing that it provides a baseline option for constructing the final block even without external builder participation.

To incorporate builder algorithms, the mechanism operates in two stages.
First, it runs the baseline algorithm to compute a baseline block and each order flow's refund.
Second, it runs any builder-submitted algorithms to determine the winning block and the value allocated to builders and the proposer.  Specifically, the mechanism generates all candidate blocks, including those generated by builder-submitted algorithms and the baseline block, and selects the one with the highest total payment as the final block. If a builder-submitted block wins, the additional surplus it offers over the second-highest block is refunded to the winning builder. All the remaining value---calculated as the winning payment minus the amounts allocated to searchers and builders---goes to the proposer. 

In particular, before these steps, the mechanism separates out transactions that do not interact with the read/write storage touched by others (which we call \emph{conflict-free}), and always appends them to the winning block. The refund for such conflict-free order flows is simply their bid in the final block, resulting in zero net payment. 

Under this design, builders' dominant strategy is always to bid their true block profit. Moreover, the VCG-based mechanism ensures that when no external builder participates or when the baseline algorithm outperforms all builders, truthful bidding is also dominant for searchers. For searchers with conflict-free order flows, bidding truthfully remains dominant even when builders may win. Finally, when a searcher has the technical ability to integrate with a builder under \name, the strategy of \emph{not} integrating and bidding truthfully still dominates any (joint) deviation for conflict-free order flows.

\parhead{The baseline algorithm}
Our theoretical analysis of the mechanism assumes a set of properties that the baseline algorithm satisfies. 
We propose a concrete algorithm as follows.

The idea is to group transactions according to their read/write conflicts, so that transactions in different groups do not interfere with each other. The remaining task is to identify the best sub-block within each group.
To assess whether the conflicts within each group are easy to resolve, we conduct an empirical study on Ethereum transaction conflicts. Specifically, we group transactions based on their conflicts over Ethereum storage and find that, in most blocks, conflict groups are small, with only a few groups containing multiple transactions. This finding is consistent with the observation that MEV opportunities are limited in most blocks.

Our empirical study shows that the idea is practical. Most conflict groups are small enough to allow full enumeration.
In special cases, such as sandwich attacks, only one bundle can be chosen, so a single pass suffices, which can reduce complexity.
For the remaining cases, we reduce the group size to a manageable subset, keeping enumeration feasible.
Our evaluation compares the baseline algorithm against the block-building algorithms provided by rbuilder~\cite{flashbots2025rbuilder}.
The results show that the baseline algorithm constructs the highest-value block among these evaluated algorithms in \betterratio of cases, while running within a practical block-building time budget.

\parhead{Contributions} We summarize our contributions as follows:
\begin{itemize}[leftmargin=*]
    \item \name, a block-building architecture that explicitly considers all three practically important actors (searchers, builders, and proposers) and separates the process of transaction collection and algorithm execution.
    \item \mechanism, a mechanism that achieves some incentive compatibility guarantees for players in \name. It also rewards builders whenever their algorithms win, thereby incentivizing innovation on advanced block-building algorithms.
    \item We present a concrete baseline block-building algorithm based on empirical insights from real-world transactions. Our evaluation shows that the algorithm is optimal in \betterratio of the cases, demonstrating its effectiveness in practice.

\end{itemize}

%% file: sections/2-background.tex
\section{Background}
\label{sec:background}

\subsection{Ethereum and Block Building}
\label{sec:ethereum-block-building}

This section provides essential background on the block building process in Ethereum, also serving as a glossary.

\parhead{Ethereum consensus}
In Ethereum Proof-of-Stake (PoS)~\cite{buterin2020combining}, validators participate by staking ETH, and time is divided into 12-second \textit{slots}. 
In each slot, one validator is pseudorandomly selected as the \textit{proposer} to construct and broadcast a block.

\parhead{Maximal Extractable Value (MEV)}
MEV refers to the value that block producers can obtain by including, excluding, or reordering transactions within a block~\cite{daian2020flash}.
MEV extraction often requires significant capital, computational resources, and advanced algorithms, which can give large players an advantage and reinforce centralization risks~\cite{yang2024decentralization,wang2026mev,yang2026geographical}.

\parhead{Proposer-builder separation (PBS) and MEV-Boost}
An important motivation behind PBS is to mitigate the centralization effects of MEV~\cite{ethereum2025pbs}. 
The idea of PBS is that proposers outsource their block-building tasks to a specialized entity called a \textit{builder}.
By auctioning off block-building rights to the external builder market, proposers gain more decentralized access to MEV revenue, since both large and small validators can benefit from competition among builders.
The current implementation of PBS is MEV-Boost~\cite{flashbots2025mevboost}.

\parhead{Searcher and bundle}
A \textit{searcher} is an independent participant in the blockchain ecosystem who monitors transactions, either from the public mempool or various private channels they can access, to identify profitable opportunities. These are commonly referred to as {\em MEV opportunities} and include arbitrage~\cite{wang2022cyclic,mclaughlin2023large,heimbach2024nonatomic,oz2025cross}, liquidations~\cite{qin2021empirical}, and sandwich attacks~\cite{zhou2021high}. To capture these MEV opportunities, searchers construct \textit{bundles}---sequences of one or more transactions---and submit them to builders, along with bids that specify how much value they are willing to share with builders.

\parhead{Block building}
A builder collects available transactions from different sources, including the public mempool as well as direct submissions of transactions or bundles from users and searchers.
We refer to these profitable transaction sets as \textit{order flows}. The builder then executes a block-building algorithm that selects among the available order flows and packages the selected ones into a valid block.

\parhead{Integration} A builder can gain a competitive advantage if it is the sole recipient of profitable bundles from a searcher, since such exclusive access increases its bidding power and therefore the joint profit of the builder and searcher pair involved. This phenomenon is called \textit{integration}~\cite{gupta2023centralizing,yang2024decentralization,oz2024wins}.

Note that not all MEV opportunities create equally strong incentives for integration. Integration is particularly attractive when a bundle is dominantly profitable, such that sending it exclusively to one builder can allow that builder to win while ensuring the bundle's inclusion. When bundle values are low or competing bundles offer comparable value, searchers may instead benefit from submitting their bundles to multiple builders to maximize their chances of inclusion.

A well-known integration pattern exists between Banana Gun and the builder Titan~\cite{yang2024decentralization,oz2024wins}. Banana Gun aggregates swap transactions targeting the same meme-coin liquidity pool and shares the resulting bundle exclusively with Titan. Due to Banana Gun's large volume, Titan can gain a dominant position in the MEV-Boost auction.
Since transactions in such niche pools are mostly generated through Banana Gun, with few user-initiated swaps, Banana Gun's bundle typically captures all transactions in these pools.
Thus, its storage accesses usually do not intersect with those of other bundles.
\done%

\subsection{VCG Mechanism}\label{sec:vcgback}
The Vickrey--Clarke--Groves (VCG) mechanism is a central achievement in mechanism design theory for selecting outcomes that maximize overall welfare based on valuation reported by strategic agents. Mapping to block building, a transaction $T$ may yield different values (MEV) to its sender depending on the blockchain state in which $T$ executes; thus, the sender can be viewed as a strategic player with a valuation function $v_T(\cdot)$ over blockchain states. The VCG mechanism can solicit truthful valuations (i.e., the reported $v_T(s)$ is exactly the value generated by executing $T$ at state $s$) and select the outcome that maximizes the overall welfare.

We recall the VCG mechanism.
Let there be $n$ agents, each with a valuation function $v_i(\omega)$ over outcomes
$\omega$ in some feasible set $\Omega$. The VCG mechanism selects the welfare-maximizing allocation
$\omega^* \in \arg\max_{\omega \in \Omega} \sum_{i=1}^n v_i(\omega)$. Each agent $i$ pays
the externality they impose on others:
$$
p_i = \max_{\omega \in \Omega} \sum_{j \ne i} v_j(\omega)
      -\sum_{j \ne i} v_j(\omega^*).
$$

Intuitively, each participant pays for ``the harm'' they cause others: $p_i$ equals the difference between (a) the maximum welfare others could achieve if $i$ were absent and (b) the welfare others actually receive under the chosen outcome $\omega^*$. In this way, any agent $i$'s utility is
$
v_i(\omega^*)-p_i=\sum_j v_j(\omega^*) - \max_{\omega \in \Omega} \sum_{j \ne i} v_j(\omega),
$
which is the overall welfare under the selected outcome minus an amount independent of $i$'s report. Thus, the maximizer of agent $i$'s utility is also the maximizer of the overall welfare.
Agent $i$ is incentivized to report truthfully, as misreporting may cause the mechanism to fail to select the maximizer.
Reporting one's true valuation is the best strategy regardless of others' reports, a property known as \emph{dominant-strategy incentive compatibility (DSIC)}.

The VCG mechanism provides a unifying framework for truthful resource allocation across diverse settings. 
In our work, VCG serves as the conceptual starting point; however, a naive implementation is computationally infeasible.

%% file: sections/3-model.tex
\section{Overview and Modeling}

\subsection{System Overview}

To situate our \mechanism mechanism in a concrete setting, we first present the \name system, whose central part is a trusted execution environment (TEE) executing the logic of \mechanism, as shown in~\autoref{fig:overall}.

\parhead{Threat model and trust assumptions}
There are three types of parties in \name: proposers, builders, and searchers.
We assume that the underlying blockchain consensus is secure.
For the main mechanism analysis, once a proposer chooses to use \name, it follows the protocol and accepts \name's output; we discuss proposer adoption in~\autoref{sec:adoption}.

The adversarial behavior in scope is strategic behavior by searchers and builders.
Searchers may misreport bid functions and integrate with builders.
Builders may likewise misreport block values, submit strategic block-building algorithms, and integrate with searchers.
These behaviors are in scope when they use the specified \name interface, even if the submitted bundles or algorithms are chosen adversarially to improve a party's payoff or reduce another party's payoff.
These deviations require only ordinary access to the submission interfaces and, for integration, out-of-band coordination.
We assume that \name rejects malformed inputs and invalid outputs that violate the interface specification.

We do not model exploits of implementation bugs in the execution client, sandbox, or TEE runtime.
TEE compromise, side-channel leakage, and host-level denial-of-service (DoS) attacks are also outside our formal model.
To focus on the mechanism design of \name, we assume an ideal attested execution environment that provides confidentiality, integrity, and isolation for mechanism execution.
This assumption is consistent with production systems such as BuilderNet, which deploy TEEs in cloud environments to reduce physical attack risk~\cite{rezabek2025proof}, and does not trivialize our design problem.

\parhead{\name}
Our system, called \name, is designed for proposers to construct and propose blocks, similar to MEV-Boost. 
The system can be run locally by a proposer who has access to TEEs (in which case other users need a way to route their inputs to proposers), or hosted as a cloud service for proposers to access, similar to the current BuilderNet. 

At the beginning of each slot (\autoref{sec:background}), \name executes two stages.
First, the system receives private order flows from users and searchers through an RPC service inside the TEE, while collecting public order flows from the public mempool.
Valid order flows are marked as available for block building, and users and searchers specify their payment for inclusion when submitting transactions.
In parallel, the system accepts block-building algorithms submitted by builders.
To control resource-exhaustion attacks, \name can use a reputation system to limit algorithm submissions, similar to searcher reputation mechanisms in existing systems~\cite{flashbots2025searcherreputation}.

At a preconfigured deadline, \name starts to produce a block, following the mechanism to be described in~\autoref{sec:mechanism}.
Roughly, \name runs both the builder-submitted algorithms and our baseline algorithm (\autoref{sec:implementation}) to build a block and to determine the refunds for builders and searchers.
Builder-submitted algorithms are not given raw system privileges.
They run in isolated sandboxes and interact with \mechanism only through the input/output interface specified above.
This prevents them from altering submitted transactions, inspecting other algorithms, communicating during execution, or leaking received order flow externally.
These restrictions remove the real-time information channel on which reactive sandwich and imitation attacks rely.
As a result, any attempt to front-run through a submitted algorithm would have to be performed without learning the protected order flow during execution, reducing it to a blind pre-submission strategy based on ex ante guesses.
The incentive compatibility and economic efficiency (i.e., the block value) are determined by the mechanism and the baseline algorithm, which is the focus of subsequent sections.

\subsection{Game Theoretic Model}
\label{sec:models}
This section establishes the game-theoretic model used to specify and analyze \mechanism.
For simplicity, we refer to both users and searchers collectively as \emph{searchers}, regarding the single transaction as a special type of bundle. As mentioned in the threat model, searchers and builders are the strategic players of interest in our game-theoretic model.  

Let $M$ denote the set of $m$ bundles submitted by searchers and $\Omega$ be the set of all feasible blocks that can be constructed using bundles in $M$. Concretely, one could think of $\Omega$ as all permutations of any subset of bundles. 
For each block $\omega \in \Omega$, the owner of each bundle $i\in M$ has a private, non-negative valuation $v_i(\omega)$, representing the value that the bundle derives from inclusion in a particular block. 
Each bundle $i$ also specifies a bid function $b_i$ of $s_i(\omega)$, which is $i$'s execution state in block $\omega$. Here, a bid $b_i\left(s_i(\omega)\right)$ represents the amount that bundle $i$ is willing to pay to a builder 
if included in block $\omega$. If bundle $i \notin \omega$, we assume that $b_i\left(s_i(\omega)\right) = 0$. 

We note that the bid $b_i(s)$ can be a constant (e.g., a fixed tip as long as it is included in the block), or it may depend on the resulting execution state $s$. The latter abstracts the common practice where searchers determine their bids to builders based on the balance change from inclusion (which reflects the realized revenue), thus paying differently when being placed in different states, as shown in this example.

\begin{example}[State-dependent proportional bid]
    This example bundle $i$ is taken from bundles following Flashbots' simple-blind-arbitrage strategy~\cite{flashbots2023bot}. 
    The relevant execution state $s_i(\omega)$ is the pair of pool reserves after all transactions in $\omega$ preceding $i$ have executed. Given this state, bundle $i$ computes a trade size from the reserves, performs two swaps, and measures revenue as the increase in its WETH balance. 
    If the revenue is non-positive, the transaction reverts, and no payment is made. Otherwise, it pays a configurable fraction (e.g., 80\%) of that realized revenue to a designated recipient $\texttt{block.coinbase}$ (which could be a builder or a TEE-controlled address in our context), with the remainder kept by the searcher. 
    In this case, the searcher's bid can be expressed as $b(s) =0.8v$ if executing bundle $i$ at state $s$ increases the searcher's WETH balance by $v$, and $b(s) =0$ otherwise.
\end{example}

Given a block $\omega$, the execution state of bundle $i$, $s_i(\omega)$, is determined. For notational convenience, we abbreviate $i$'s bid in block $\omega$, $b_i\left(s_i(\omega)\right)$, simply as $b_i(\omega)$ throughout the rest of the paper. 
We use $b_i$ to denote the function when the context is clear. The collection of all bid functions is denoted by $\mathbf{b}$, while $\mathbf{b}_{-i}$ refers to the bid functions of all bundles except $i$.

Let $N$ be the set of $n$ block-building algorithms submitted by builders. Each algorithm $j\in N$ takes the set $M$ of bundles as input and outputs a candidate block $B_j$, along
with a bid $\beta_j\in \mathbb{R}_{\geq0}$, representing the amount the builder is willing to pay to the mechanism \mechanism if $B_j$ is selected.

To capture the practical form of integration, where a searcher shares their bundle exclusively with a single builder, we formalize integration in our model as follows.
\begin{definition}
    A searcher $i$ is said to integrate with builder $j$ if the bundle executes as intended and produces its intended value when builder $j$'s algorithm is run by the mechanism, but reverts and yields zero value under other algorithms.
\end{definition}

\begin{remark}[Integration under \name]\label{remark:integration}
Looking ahead, our mechanism, \mechanism, enforces that every builder-submitted algorithm, as well as the baseline algorithm, operates on the same set $M$ of bundles. Nevertheless, a strategic searcher can still implement integration through a state-dependent execution trick. 
Concretely, an integrated searcher can construct a bundle that executes only if 
the coinbase address~\cite{wood2014ethereum} equals a predetermined target. Otherwise, the bundle reverts and pays nothing. 
When the integrated algorithm is executed by \mechanism, it sets the \texttt{block.coinbase} to the integrating builder's address (i.e., the target address), thereby ``signaling'' the bundle to execute its intended logic; under all other algorithms, the bundle behaves as a no-op. 
In this way, the same bundle exhibits different behaviors depending on which algorithm is run, allowing searcher-builder integration to re-emerge even under~\name.
\end{remark}

We emphasize that, from the perspective of any non-integrating builder, the integrating bundle $i$ behaves as \emph{empty} or \emph{no-op}: for any blockchain state $s$, executing $i$ leaves the state unchanged, i.e., $s(i) = s$. 
This is strictly stronger than just bidding zero; non-integrating builders cannot benefit from the state change by including the bundle because it performs no execution.

Moreover, this integration issue cannot be easily resolved by hiding \texttt{block.coinbase} during execution, since such a solution can also be bypassed: the bundle executes its intended logic only under a specific on-chain state, which is triggered by a particular execution ordering known to the integrated builder.
More broadly, given the basic property that smart contracts can execute differently depending on the on-chain state, we cannot fully eliminate this issue.
What we can do is make such integration harder, in the sense that a bundle can identify its integrated algorithm during execution, but cannot distinguish the baseline algorithm from other algorithms.

\subsection{Block-building Mechanism}\label{subsec:Block-BuildingMechanism}
The input to a {\em block-building mechanism} consists of a set $M$ of bundles, their bidding functions $\mathbf{b}$, and a set of $n$ builder-submitted algorithms. The output is a block and the net payment of each participant (a player may pay or get paid under the mechanism). 
Technically, the payments are settled by transactions in the output block.

\parhead{A note on how we specify payments}
Our mechanism implements payment rules by charging an upfront payment, then \emph{refunding} a portion of it, to be compatible with how blockchain transactions work. From a mechanism-design point of view, only the {\em net} payment matters. 
In \name, algorithms running in the TEE can determine refund amounts and embed refund transactions in the final block.

In order to formally analyze the participants' incentives under this refund-based design, we next define the utility functions for searchers and builders.

\begin{definition}[Searcher Utility Function]
    Given a mechanism outcome in which block $\omega \in \Omega$ is selected, the utility of searcher $i \in M$ who bids according to function $b_i$ and receives refund $r_i \in \mathbb{R}_{\geq 0}$ from the mechanism is
    $$u_i^{\text{searcher}} = v_i(\omega) - b_i(\omega) + r_i.$$
\end{definition}

\begin{definition}[Builder Utility Function]
    For an instance $M$ with the bidding vector $\mathbf{b}$, given a mechanism outcome in which builder $j$'s algorithm produces block $B_j$ and bid $\beta_j$, and the mechanism selects block $\omega$ as the winning block, the utility of builder $j$ who receives refund $r_j$ is
    $$
    u_j^{\text{builder}} =
    \begin{cases}
    \sum_{i\in M} b_i(\omega) - \beta_j + r_j, & \text{if } B_j = \omega; \\
    0, & \text{otherwise}.
    \end{cases}
    $$
    Here, $\sum_{i\in M} b_i(\omega)$ represents the builder's gross revenue from the included bids if builder $j$'s block is selected (i.e., $B_j = \omega$). Builders incur no cost if their block is not selected.
\end{definition}

In an ideal scenario, one might aim for a mechanism that satisfies the following properties:  

\begin{itemize}[leftmargin=*]
    \item \textit{Searcher-DSIC}: For any searcher who chooses non-integration, truthfully reporting their private valuation is a dominant strategy, i.e., $b_i(\omega) = v_i(\omega)$ for any block $\omega$. \done%
    \item \textit{Builder-DSIC}: Truthfully reporting their private valuation is a dominant strategy for builders. Note that even if a builder is integrated with some searcher, reporting is still needed and the property requires that truthfully reporting remains a dominant strategy in this case.
    \item \textit{Integration-resistance}: No searcher–builder pair can increase their joint utility by integration.
    \item \textit{Welfare maximization}: The mechanism chooses a block that maximizes the total realized value (sum of searcher valuations) subject to feasibility constraints.
    \item \textit{Computational efficiency}: The mechanism runs in polynomial time in $m$ and $n$.
    \item \textit{Individual rationality (IR)}: Searchers and builders obtain non-negative utility under truthful participation. 
    \item \textit{Budget balance}: The total payments collected by the mechanism are no less than the total money paid out by the mechanism; hence, the mechanism never needs subsidies.
\end{itemize}

\begin{remark}[Integration-resistance vs. Searcher-DSIC]
Integration-resistance emphasizes the equal visibility of the bundle's intended logic, while Searcher-DSIC requires truthful reporting of the realized valuation. 
Importantly, when a searcher integrates, its bundle becomes invisible to all other algorithms, let alone the bidding function, so there is no meaningful notion of truthfulness.
Accordingly, we discuss Searcher-DSIC when the searcher chooses not to integrate.
\end{remark}

In our setting, achieving all properties simultaneously is infeasible. 
The block-building task includes \emph{which} bundles are included and in \emph{what} order, so the size of full block space, $|\Omega|$, grows exponentially with $m$ (the number of bundles). Thus, even if the DSIC property holds, deciding the welfare-maximizing block subsumes classic $\mathsf{NP}$-hard problems~\cite{Mikerah2024knapsack}, implying that the \emph{exact} welfare maximization is computationally intractable unless $\mathsf{P} = \mathsf{NP}$. 

While DSIC and welfare maximization are central objectives in mechanism design, computational efficiency reflects a strict practical constraint in block building. For example, a mechanism must produce a block within about 12 seconds in Ethereum to remain useful for proposers. 
Our ultimate goal, therefore, is to design a mechanism that (1) is computationally efficient in practice, (2) provably satisfies DSIC in a meaningful set of cases, (3) provably satisfies integration-resistance in the practically concerning scenarios, and (4) achieves welfare maximization as effectively as possible in practice.

%% file: sections/4-mechanism.tex
\section{Mechanism}
\label{sec:mechanism}

This section presents our mechanism in three steps.
~\autoref{subsec:idea} considers an idealized case where the TEE is assumed to have ``unbounded'' computational resources, showing that a direct implementation of the VCG mechanism achieves a desirable outcome (DSIC and welfare maximization), but it is computationally infeasible.
Thus, \autoref{subsec:practical} moves to a scenario with limited computational resources, 
and presents a tractable VCG variant. This mechanism remains DSIC but may lose some welfare due to the computational power constraint.
Finally, \autoref{subsec:mechanism} extends the design by incorporating builder-submitted algorithms to potentially further improve the welfare.

\subsection{TEE with Oracle Access: Direct VCG}\label{subsec:idea}
We begin with an idealized setting in which the TEE is augmented with an oracle, which is, given any bid functions $\mathbf{b}$ from the searchers, capable of returning a block $\omega^* \in \Omega$ that maximizes the sum of all searchers' bids. 
Here, recall that $\Omega$ is the set of all feasible blocks that can be constructed from the bundles in $M$. With such unbounded computational power, the TEE can directly apply the VCG mechanism, which ensures DSIC for searchers. %
We formally introduce the VCG mechanism in our setting below: %
\begin{enumerate}
    \item[(1)] Query the oracle to obtain an optimal block $\omega^*$ that maximizes the total bid, namely, $\omega^* \in \arg \max_{\omega\in \Omega} \sum_{i\in M}b_i(\omega)$, and outputs it as the final block;
    \item[(2)] Charge each bundle $i$ the amount $b_i(\omega^*)$;
    \item[(3)] For each bundle $i$, query the oracle to compute $\max_{\omega\in \Omega} \sum_{j\neq i} b_j(\omega)$,\footnote{Note that this is equivalent to querying the oracle with the bid function $\mathbf{b}$ except that $\mathbf{b_i}$ is changed to zero.} calculate the refund: 
    $$r_i(\mathbf{b}) = \sum_{j\in M} b_j(\omega^*) - \max_{\omega\in \Omega} \sum_{j\neq i} b_j(\omega),$$
    and refund $r_i(\mathbf{b})$ back to the owner of bundle $i$;
    \item[(4)] Transfer the remaining amount $\sum_{i\in M} \left(b_i(\omega^*) - r_i(\mathbf{b}) \right)$ to the proposer.
\end{enumerate}

We note that in step (1), the mechanism only maximizes the \emph{reported welfare} with respect to searchers' bidding functions. However, since the mechanism is DSIC, truthful bidding ensures that the selected block $\omega^*$ also maximizes the sum of searchers' true valuations, i.e., social welfare. 

Combining steps (2) and (3), bundle $i$'s \emph{net} payment is 
$$b_i(\omega^*) - r_i(\mathbf{b}) = \max_{\omega\in \Omega} \sum_{j\neq i} b_j(\omega) - \sum_{j\neq i} b_j(\omega^*),$$
which is the difference between the maximum welfare others could achieve if $i$ were absent and the welfare others actually receive under the chosen outcome $\omega^*$. 

\begin{example}\label{ex:vcg}
    Consider two bundles whose bids depend on whether they execute at the head of the block (i.e., on the state immediately after the previous block) or later. 
    \begin{itemize}
        \item Bundle 1: bids 40 if executed first; 100 if included later.
        \item Bundle 2: bids 50 if executed first; 80 if included later.
    \end{itemize}

    \begin{table}[h!]
        \centering
        \caption{Bids under different blocks.}
        \begin{tabular}{c|c|c|c}
            \hline\hline
            Block $\omega$    &   $b_1(\omega)$   &   $b_2(\omega)$   &   Total Bid   \\ \hline
            [1]         &   40      &   0       &   40       \\ \hline
            [2]         &   0       &   50      &   50      \\  \hline
            [1, 2]      &   40      &   80      &   120      \\ \hline
            [2, 1]      &   100     &   50      &   150      \\ \hline \hline
        \end{tabular}
        \label{table:ideal}
    \end{table}
    
    In this example, the optimal block is $\omega^* = [2, 1]$ with a total bid of 150. To compute refunds,
    \begin{itemize}
        \item For bundle 1: the maximum total bid without $b_1$ is $\max_{\omega\in \Omega} \sum_{j\neq i} b_j(\omega) = \max_{\omega\in \Omega} b_2(\omega) = b_2([1,2]) = 80$. Thus $r_1(\mathbf{b}) = 150 - 80 = 70$.
        \item For bundle 2: the maximum total bid without $b_2$ is $\max_{\omega\in \Omega} \sum_{j\neq i} b_j(\omega) = \max_{\omega\in \Omega} b_1(\omega) = b_1([2,1]) = 100$. Thus $r_2(\mathbf{b}) = 150 - 100 = 50$.
    \end{itemize}
    Overall, the TEE proposes [2, 1] as the final block and transfers 30 to the proposer.
\end{example}

\begin{theorem}[\cite{vickrey1961counterspeculation,clarke1971multipart,groves1973incentives}]
    It is a dominant strategy for each bundle to truthfully report its value (namely, setting $b_i = v_i$).
\end{theorem}

\subsection{Efficient VCG-based Mechanism}\label{subsec:practical}
The VCG mechanism works well except for the assumption of the ideal oracle. In practice, the size of $m$ is in the hundreds or thousands, causing simulating all feasible blocks in $\Omega$ to be computationally intractable. Instead, in this section, we propose a framework for designing mechanisms that remain DSIC for searchers and run efficiently in practice. 

The main idea is to design VCG-based mechanisms, which preserve the steps (2) -- (4) of the VCG mechanism while replacing the oracle in step (1) with an efficient algorithm denoted by $\mathcal{A}$. 
The algorithm $\mathcal{A}$ takes the set $M$ of bundles and their bids $\mathbf{b}$ as inputs and outputs $m+1$ blocks introduced below, where $m$ is the number of bundles in $M$. 

To reach the goal of running efficiently and maintaining truthfulness, we require the algorithm $\mathcal{A}$ to satisfy certain key properties, formalized in~\autoref{property}. 
At a high level, such an algorithm $\mathcal{A}$ operates on a subset of blocks $\Omega_\mathcal{A} \subseteq \Omega$, which must be independent of the bid functions. Moreover, within this subset, the algorithm should be able to efficiently find a block $o^* \in \Omega_\mathcal{A}$ that maximizes the total bid. 
The extent to which the mechanism can approximate the optimal welfare ultimately depends on the performance of $\mathcal{A}$.

\begin{property}\label{property}
The desired properties that an algorithm $\mathcal{A}$ should satisfy are as follows.
    \begin{enumerate}[leftmargin=*]
    \item[(i)] For any set $M$ of bundles, the set of blocks $\Omega_\mathcal{A}$ considered by $\mathcal{A}$ does not depend on bundles' bids. Formally, for all $M$ and for all bid profiles $\mathbf{b}, \mathbf{b'}$, $\Omega_\mathcal{A}(M, \mathbf{b}) = \Omega_\mathcal{A}(M, \mathbf{b'})$.  
    \item[(ii)] $\mathcal{A}$ can compute the following in time polynomial in $m$: 1) $o^* \in \arg \max_{o\in \Omega_\mathcal{A}} \sum_{i\in M} b_i(o)$, and 2) for each $i\in M$, the counterfactual optimal block $o_{-i} \in \arg \max_{o\in \Omega_\mathcal{A}} \sum_{j\neq i} b_j(o)$.
\end{enumerate}
\end{property}

Condition (i) means that the feasible set $\Omega_\mathcal{A}$ is independent of bids; this ensures that reducing the feasible set does not affect bidders' strategic bidding and, consequently, the truthfulness analysis.
Condition (ii) emphasizes optimality and computational feasibility after reducing the feasible set.
Looking forward, in~\autoref{sec:implementation}, we present a practical block building algorithm that satisfies these properties.

\parhead{VCG-based Mechanism} 
With an algorithm $\mathcal{A}$ satisfying the above properties, we can construct the following VCG-based mechanism, denoted by $\texttt{VCG}_\mathcal{A}$:
    \begin{enumerate}
        \item[(1)] Run $\mathcal{A}$ to get  ``optimal'' blocks $o^*$ and $\{o_{-i}\}_{i\in M}$, and output $o^*$ as the final block;
        \item[(2)] Charge each bundle $i$ the amount $b_i(o^*)$;
        \item[(3)] Refund each bundle $i$:  
        
        \begin{equation}
            r_i(\mathbf{b}) = \sum_{j\in M} b_j(o^*) - \sum_{j\neq i} b_j(o_{-i});
            \label{eq:refundfunction}
        \end{equation}
        \item[(4)] Transfer the remaining amount $\sum_{i\in M} \left(b_i(o^*) - r_i(\mathbf{b}) \right)$ to the proposer.
    \end{enumerate}

The following theorem shows that $\texttt{VCG}_\mathcal{A}$ preserves truthfulness as long as the algorithm $\mathcal{A}$ satisfies~\autoref{property}.

\begin{theorem}\label{theorem:variantDSIC}
    For any algorithm $\mathcal{A}$ that satisfies conditions (i) and (ii), the resulting mechanism $\texttt{VCG}_\mathcal{A}$ operates efficiently and is DSIC for searchers.
\end{theorem}

\emph{Proof sketch.}
Condition~(ii) gives efficient execution, while condition~(i) ensures the feasible set is independent of any searcher's report.
Thus, their bids affect only which block is selected from a fixed set, and truthful reporting aligns the searcher's utility-maximizing block with the block selected by $\mathcal{A}$.
The full proof appears in~\autoref{app:proof-variantDSIC}.

\subsection{Our Mechanism with Multiple Builders}\label{subsec:mechanism}

The mechanism $\texttt{VCG}_\mathcal{A}$ has searcher-DSIC and computational efficiency properties.
However, it uses a single fixed block-building algorithm $\mathcal{A}$. 
One natural next step is to elicit and leverage builders' additional sophistication by asking them to submit their own block-building algorithms. To differentiate, we refer to $\mathcal{A}$ as the {\em baseline algorithm} of our mechanism, which, as we will show shortly, plays a crucial role in achieving searcher-DSIC.

\parhead{Builder-submitted algorithms}
A builder-submitted algorithm takes the set $M$ of bundles and their bidding functions $\mathbf{b}$
as input, and outputs a block and a bid value to be paid to our mechanism. 
Each algorithm is submitted privately by a builder, and the mechanism executes it on the common set of received bundles to generate candidate blocks that compete with the block produced by the baseline algorithm. 

The mechanism also checks that every algorithm's output is a permuted subset of the input (i.e., the algorithm does not insert new transactions).
This constraint is not only necessary to prevent covert channels through which builders could send exclusive bundles to their algorithms, but also reflects current practice.
In practice, builders are cautious when inserting transactions, as arbitrary insertion would visibly signal builder power and may make searchers fear targeted attacks. As a result, existing builders typically claim neutrality and limit insertion to a single case: subsidizing transactions that cannot pay fees upfront but yield higher payments if executed~\cite{titan2023teatime}.
Enforcing permutation-only outputs therefore aligns with both security goals and real-world builder practices.

\begin{figure*}[!ht]
\centering
\begin{tcolorbox}[colframe=black, colback=white, width=0.95\textwidth, boxrule=0.8pt]
\small
\textbf{Input:} A set $M$ of bundles, their bidding functions $\mathbf{b}$, and a set of $n$ builder-submitted algorithms. Each algorithm $j$ takes $M$ as input and outputs a candidate block $B_j$ along with a bid $\beta_j$.

\textbf{Output:} A final block for the proposer and an allocation of the block's total value.

\vspace{0.3em}
\textbf{Mechanism:}
\begin{enumerate}
    \item[(0)] Identify all \emph{conflict-free} bundles $S \subseteq M$. 
    \item[(1)] Run the baseline algorithm $\mathcal{A}$ on the remaining bundles $M \setminus S$ to generate: 
    \begin{enumerate}
        \item[(i)] an optimal block $o^* \in \arg \max_{o\in \Omega_\mathcal{A}} \sum_{i\in M\setminus S} b_i(o)$ as the baseline block, with value $\beta_0 \coloneqq\sum_{j\in M\setminus S} b_j(o^*)$;
        \item[(ii)] a ``counterfactual optimal'' block $o_{-i} \in \arg \max_{o\in \Omega_\mathcal{A}} \sum_{j\neq i} b_j(o)$ for each $i\in M \setminus S$.
    \end{enumerate}
    
    \item[(2)] For each bundle $i\in M\setminus S$, compute the refund $r_i(\mathbf{b}) = \sum_{j\in M\setminus S} b_j(o^*) - \sum_{j\neq i} b_j(o_{-i})$.

    \item[(3)] Run all builder-submitted algorithms on the bundle set $M \setminus S$ to obtain $n$ candidate blocks $\{B_1, \cdots, B_n\}$ and their bids $\{\beta_1, \cdots, \beta_n\}$. 
    Let $B^*$ be the block with the highest bid $\beta^* = \max_{j\in [n]} \{\beta_j\}$ (\textit{arbitrary tie-breaking}), and $B'$ be the block with the second-highest bid $\beta' = \max \{ \beta_j \mid j\in [n] \text{ and } j \neq \text{builder of } B^* \}$.
    \item[(4)] Determine the final output as follows:
    
    \textbf{Case (4a).} $\beta_0 \geq \beta^*$ (baseline algorithm wins). 
        \begin{itemize}
            \item[-] Final block: $o^* \gets o^* \cup S$, i.e., append the conflict-free bundles to the baseline block. %
            \item[-] Builder payments: All builders pay nothing.
            \item[-] Searcher payments: Each bundle $i\in M$ pays $b_i(o^*)$ to the mechanism. The mechanism refunds $r_i$ to each bundle $i\in M \setminus S$ and refunds $b_i(o^*)$ to each bundle $i\in S$.
            \item[-] Proposer revenue: The proposer receives the remaining value $\beta_0 - \sum_{i\in M \setminus S} r_i$.
        \end{itemize}
        
        \textbf{Case (4b).} $\beta_0 < \beta^*$ (a builder wins).
        \begin{itemize}
            \item[-] Final block: $B^* \gets B^* \cup S$, i.e., append the conflict-free bundles to the winning builder's block.
            \item[-] \label{itm:4b-refund} Builder payments: The winning builder pays $\beta^* + \sum_{i\in S} b_i(B^*)$ to the mechanism, which then refunds it the surplus $\beta^* - \max\{\beta_0, \beta'\}$. (All other builders pay $0$.)%
            \item[-] Searcher payments: Each bundle $i\in M$ pays $b_i(B^*)$ to the winning builder. The mechanism refunds $r_i$ to each bundle $i\in M \setminus S$ and refunds $b_i(B^*)$ to each bundle $i\in S$.
            \item[-] Proposer revenue: The proposer receives the remaining value $\max\{\beta_0, \beta'\} - \sum_{i\in M \setminus S} r_i$.
        \end{itemize}
\end{enumerate}
\end{tcolorbox}
\caption{Our mechanism \mechanism}
\label{fig:mechanism}
\end{figure*}

\parhead{Our mechanism \mechanism}
As specified in~\autoref{fig:mechanism}, \mechanism is largely based on $\texttt{VCG}_\mathcal{A}$, extending it by adding step~(3), Case~(4b), and a preliminary step~(0). 
Conceptually, the mechanism proceeds in two phases. 
In the first phase, including steps (1) and (2), the mechanism runs the baseline algorithm to obtain a baseline block $o^*$ and compute the refund $r_i$ to each bundle $i\in M$. 
In the second phase, including steps (3) and (4), the mechanism runs all builder-submitted algorithms to generate additional candidate blocks, and determines the winning block and the value distributed to builders and the proposer. 
After committing to the inputs and final block, \mechanism can compute searcher refunds asynchronously and settle them in batches in subsequent blocks.

Before these two phases, \mechanism first identifies the set $S \subseteq M$ of all \emph{conflict-free} bundles in step~(0). A bundle is called \emph{conflict-free} if its execution neither affects nor depends on that of any other bundle in $M$.
In practice, conflict-free bundles can be identified by analyzing read/write storage slots accessed by each bundle.\footnote{If a bundle reads from and writes to storage slots that are disjoint from those of all others, its inclusion will not alter nor be altered by the inclusion or ordering of other bundles.}
Looking ahead,~\autoref{sec:empirical-insights} details this approach and the \texttt{GetConflictGroups} function in~\autoref{alg:default} provides an efficient procedure of grouping bundles based on their storage-level interactions. 
Here, by invoking this function and selecting groups containing only a single bundle, the mechanism can efficiently identify all conflict-free bundles $S$.

Since conflict-free bundles can be executed in any state without affecting the rest,~\mechanism handles them in two stages. 
First, it withholds these bundles $S$ and runs all algorithms on the remaining set $M\setminus S$. 
Then, it appends $S$ to each candidate block.\footnote{We assume that the base-fee adjustment rule in EIP-1559 functions effectively, and that~$M$ contains all eligible transactions willing to pay the base fee. Under this assumption, all bundles in~$M$ can fit within a single block, consistent with our empirical findings in \autoref{app:gas-usage}.} 
This update is uniform across all candidate blocks and does not affect the winner selection. Thus, for simplicity, we omit this global update from intermediate steps and apply it only to the winning block in step~(4). 
This two-stage design offers two benefits: 
\textit{(i)} Conflict-free bundles require no sophisticated ordering, so separating them reduces the algorithm's input and improves efficiency; and 
\textit{(ii)} Looking ahead, this serves as a crucial ingredient to achieve integration-resistance for conflict-free bundles, which capture one of the most concerning integration scenarios: as noted in~\autoref{sec:ethereum-block-building}, Banana Gun's bundles can avoid interacting with storage touched by others and thus serve as real-world examples of conflict-free bundles (see \autoref{app:empirical-study-of-banana-gun} for further supporting measurements). 

Among all candidate blocks from the baseline algorithm and builder-submitted algorithms, the one with the highest bid value is selected as the winning block. Here, a block's bid is defined as \mechanism's balance increase.\footnote{\mechanism can receive payments and disburse rewards; in practice, this can be done by having the TEE running \mechanism control a blockchain account.}
For the baseline block $o^*$, this value is $\beta_0 \coloneqq\sum_{j\in M\setminus S} b_j(o^*)$, i.e., total payment from included bundles to the mechanism. For a block $B_j$ produced by builder $j$'s algorithm, this is their bid~$\beta_j$. 

If the baseline algorithm wins (Case 4a), the mechanism finalizes the baseline block by appending the conflict-free bundles, i.e., $o^* \gets o^* \cup S$, and uses the resulting block as the final output. Under this complete block, each searcher $i\in M$ pays $b_i(o^*)$ to \mechanism, so the mechanism collects $\beta_0 + \sum_{i\in S} b_i(o^*)$ in total. The mechanism then allocates this value among all searchers and the proposer: each searcher $i\in M\setminus S$ receives the refund $r_i$ defined in~step (2), each searcher $i\in S$ receives their bid $b_i(o^*)$ as refund (resulting in zero net payment), and the remaining value $\beta_0 - \sum_{i\in M\setminus S} r_i$ goes to the proposer. 
Note that in this case, the mechanism essentially degenerates into the VCG-based mechanism $\texttt{VCG}_\mathcal{A}$ in~\autoref{subsec:practical}. In particular, we use the baseline algorithm to generate the final block and use the same refund function Eq.~(\ref{eq:refundfunction}) to compute the refund to searchers.\footnote{If the algorithm $\mathcal{A}$ in $\texttt{VCG}_\mathcal{A}$ includes a conflict-free bundle~$i$ in~$o^*$, then computing the refund following Eq.~(\ref{eq:refundfunction}) will yield the same value as set in Case~(4a), i.e., the bundle's own bid~$b_i(o^*)$.}
The baseline algorithm's winning implies that builder-submitted algorithms do not improve reported welfare over the baseline. Thus, the mechanism rewards nothing to builders, and transfers the remaining value to the proposer, exactly as in $\texttt{VCG}_\mathcal{A}$.

If a builder-submitted algorithm wins (Case 4b), likewise, the mechanism finalizes the winning block by appending the conflict-free bundles, i.e., $B^* \gets B^* \cup S$, and outputs it as the final block. 
Under this complete block, searchers pay a total of $\sum_{i\in M} b_i(B^*)$ to the winning builder, who in turn pays their updated bid $\beta^* + \sum_{i\in S} b_i(B^*)$ to the mechanism.\footnote{This can be implemented by requiring builders to post collateral, which can then be reduced by \name.} %
The mechanism distributes this amount among searchers, builders, and the proposer: it refunds $r_i$ to each bundle $i\in M\setminus S$ and $b_i(B^*)$ to each bundle $i\in S$, refunds the surplus between $\beta^*$ and $\max\{\beta_0, \beta'\}$ to the winning builder, and transfers the rest to the proposer. 
Intuitively, this mirrors a second-price auction among builders. The winning builder pays just enough to outbid the second-best algorithm, and pockets the surplus as profit. In this way, Case (4b) extends the VCG-based mechanism by rewarding builders for contributing algorithms that generate strictly better blocks, thus fostering innovation. %

We emphasize that in Case (4b), the refund $r_i$ to each bundle $i\in M\setminus S$ is computed solely from the \textit{baseline algorithm} in step~(2), while the refund to each bundle $i\in S$ is a fixed value independent of ordering. That means in either case, the refunds to searchers are identical, regardless of which block is ultimately selected. 
One might wonder why we do not always use the \emph{winning} algorithm to calculate the refund. While seemingly natural, such an alternative design is vulnerable: a strategic builder could collude with a searcher to inflate their refunds. Due to the page limit, we postpone the details to~\autoref{sec:alternative-refund}.

\subsection{Analysis of Our Mechanism}

We start with analyzing the refund design in \mechanism. 

\begin{claim}\label{claim:non-negative-refund}
    The refund to each bundle $i\in M$ is non-negative.
\end{claim}

\begin{corollary}[Individual Rationality of Searchers]\label{cor:searchersIR}
    Truthful searchers receive non-negative utility.
\end{corollary}

\begin{theorem}\label{thm:proposerIR}
    The total refunds to searchers and builders do not exceed the total payment received by our mechanism.
\end{theorem}

We postpone the proofs to~\autoref{app:proof-non-negative-refund} and~\autoref{app:proof-proposerIR}.

Next, we discuss the incentives of builders and searchers. We first focus on the bidding strategy, assuming that searchers and builders operate independently, i.e., without integration.
\begin{theorem}[DSIC for builders]\label{theorem:builderDSIC} 
    \mechanism is DSIC for builders.
\end{theorem}
\begin{proof}
    When deciding the winner, the mechanism selects the algorithm with the highest bid. 
    If a builder wins, the \emph{effective} payment is the maximum of the baseline block value ($\beta_0$) and the second-highest builder bid. This is precisely a second-price auction with a reserve $\beta_0$, which is DSIC.
\end{proof}

\begin{theorem}[DSIC for searchers in baseline-dominating scenario]\label{theorem:searcherDSIC}
    Given a bundle set where the baseline algorithm outperforms all builder-submitted algorithms for all bid functions, bidding truthfully is a dominant strategy for searchers. 
\end{theorem}
\begin{proof}
    For any conflict-free bundle $i \in S$,~\mechanism always appends it to the winning block and refunds its bid. Hence, $i$ has zero net payment and fixed utility, regardless of its bid. 
    For bundles in $M \setminus S$, if the baseline algorithm $\mathcal{A}$ always wins for all $M$ and all bid profiles $\mathbf{b}$, the mechanism always enters Case (4a) and coincides with the VCG-based mechanism $\texttt{VCG}_\mathcal{A}$. Therefore, the theorem naturally holds for bundles in $M \setminus S$ following~\autoref{theorem:variantDSIC}.
\end{proof}

The baseline-dominating scenario includes two cases: the baseline algorithm is the only one running in the TEE (e.g., in the initial stages after releasing our system), or the baseline algorithm performs better than all other builder-submitted algorithms (e.g., due to the community's contribution). 
In this scenario,~\autoref{theorem:builderDSIC} and~\autoref{theorem:searcherDSIC} tell us that both builders and searchers should bid truthfully. 

For the general scenario where a builder-submitted algorithm may win, there is a chance that the mechanism enters Case (4b). However, builder-submitted algorithms effectively behave as black boxes, making it difficult to obtain universal guarantees. 
In particular, a builder's algorithm might not satisfy~\autoref{property}, and therefore \mechanism loses DSIC for searchers. This highlights a fundamental tradeoff between ensuring mechanisms that are always-DSIC for searchers and harnessing the sophistication of builders.

Nevertheless, for searchers with conflict-free bundles, DSIC holds in all scenarios, even when the baseline algorithm is not dominant.  Formally, we show the following theorem. 

\begin{theorem}
    Bidding truthfully is a weakly dominant strategy for searchers with conflict-free bundles in all scenarios.
\end{theorem}

This argument follows the same reasoning as in the proof of~\autoref{theorem:searcherDSIC}. 
Intuitively, for a conflict-free bundle, neither the allocation nor the net payment depends on its bid. 
Hence, truthful bidding weakly dominates any deviation for conflict-free searchers, even in builder-winning cases.

Next, we turn to the integration decision faced by searchers. 
In particular, we focus on the searcher with a conflict-free bundle, which represents one of the most concerning practical sources of integration.

Recall that while \mechanism ensures all algorithms ``see'' the same set of bundles, an integrating bundle can still behave differently under different algorithms by employing a state-dependent execution trick (see~\autoref{remark:integration}). Crucially, such a bundle can only distinguish the integrated algorithm from all others; it cannot further differentiate among remaining non-integrated algorithms (e.g., it cannot tell the baseline algorithm apart from those submitted by other builders\footnote{To prevent the searcher from distinguishing the baseline algorithm by observing the coinbase address, the baseline algorithm uses a one-time address as the coinbase address for block building.}). 
Thus, the searcher's integration strategy space effectively reduces to two options: grant \emph{equal} access to all algorithms, or grant \emph{exclusive} access to the integrated algorithm. 

Mathematically, the searcher's two possible behaviors can be represented by the two following strategies: 
\begin{itemize}[leftmargin=*]
    \item \emph{Participate:} apply no execution trick and submit the bundle to~\name, allowing all algorithms to access it equally; 
    \item \emph{Integrate:} withhold the bundle from submission and privately send it to the integrated builder.
\end{itemize}

These two strategies capture whether the searcher restricts access to the integrated building algorithm or makes the bundle universally available. 
It is important to note that this equivalence holds only in a \emph{mathematical} sense. In practice, \name does not permit any algorithm to introduce new bundles. Thus, a bundle privately sent to the integrated builder cannot actually be included in the block, whereas the state-dependent execution trick described in~\autoref{remark:integration} effectively circumvents this restriction, enabling the integrated algorithm to include the bundle in its proposed block. 

Hence, before setting their bidding function, each searcher faces an additional \emph{meta-decision}: whether to integrate or to participate. 
Together, these choices form a two-step decision process, illustrated in~\autoref{fig:searcher-tree}. 

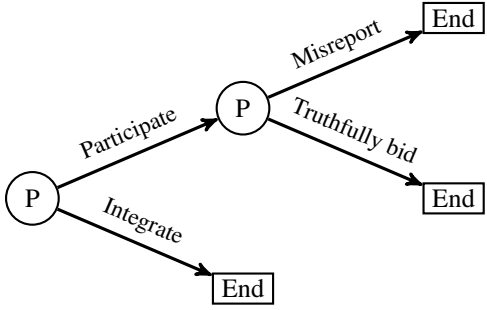
\begin{figure}[t]
  \centering
  \begin{tikzpicture}[
    scale=0.8,
      transform shape,
      >=stealth',
      grow=right,
      level distance=28mm,
      level 1/.style={sibling distance=24mm},
      level 2/.style={sibling distance=24mm},
      player/.style={circle, draw, thick, minimum size=7mm, inner sep=0pt},
      terminal/.style={rectangle, draw, thick, inner sep=1pt, minimum width=8mm, minimum height=4mm},
      edge from parent/.style={->, very thick, draw=black},
      elabel/.style={midway, above, sloped, font=\small}
    ]
    \node[player] (root) {P}
      child { node[terminal] {End}
              edge from parent node[elabel] {Integrate} }
      child { node[player] (p2) {P}
        child { node[terminal] {End}
                edge from parent node[elabel] {Truthfully bid} }
        child { node[terminal] {End}
                edge from parent node[elabel] {Misreport} }
        edge from parent node[elabel] {Participate} };
  \end{tikzpicture}

  \caption{Searcher's two-step decision: first, choose whether to integrate or participate; conditional on participation, choose between truthful bidding or misreporting. Each path leads to a terminal outcome.}
  \label{fig:searcher-tree}
\end{figure}

Regarding the meta-game, we give the following theorem.

\begin{theorem}[Integration-resistance for conflict-free bundles]\label{theorem:conflictfreeIntegration}
    Let $i\in M$ be a conflict-free bundle and let $j\in[n]$ be any builder. 
    For any bids of all builders, the strategy of searcher $i$, participate and bid truthfully, weakly dominates all alternative strategies with respect to the joint utility of~$(i, j)$.
\end{theorem}

Intuitively, under truthful participation, the mechanism already refunds the conflict-free bundle its full value. By contrast, integration can only reduce the bundle's chance of inclusion or cause the integrated builder to win while paying more than the true gain, thus lowering their joint utility. Hence, any deviation through integration or misreporting cannot improve the coalition's total utility relative to truthful participation.
The full proof appears in~\autoref{app:proof-conflictfreeIntegration}.

%% file: sections/5-implementation.tex
\section{Design and Evaluation of Baseline Algorithm}
\label{sec:implementation}

To implement \name, we need a baseline algorithm compatible with the incentive argument in \autoref{sec:mechanism}.
In particular, the algorithm must satisfy \autoref{property}: its candidate block set must be independent of submitted bids.
The challenge is to satisfy \autoref{property} while still constructing high-value blocks.
Existing block-building algorithms we are aware of, including those in \texttt{rbuilder}, use bids to guide ordering, pruning, or candidate selection, which makes their candidate block sets bid-dependent.
This gap motivates a new baseline algorithm that satisfies \autoref{property} while remaining competitive in block value.
Note that we omit the implementation and evaluation of an end-to-end system because its performance profile will be similar to existing systems such as BuilderNet.

\subsection{Design Principle}
\label{sec:empirical-insights}

For non-conflicting transactions whose execution does not affect others' bids, block building is simple: any order works, and a greedy algorithm suffices.
Conflicts arise when transactions interact, especially when multiple searchers compete for the same MEV opportunity.
In the general case, resolving conflicts requires searching over all permutations of all subsets.
Our idea is to partition transactions into {\em conflict groups} and solve each group independently.
We exhaustively select the best ordering within each group and then combine them into the final block.
As we show, this is practical because conflicts are uncommon in real-world transactions, and most conflict groups are small enough to be exhaustively searched.

\subsection{Empirical Validation}
\label{sec:empirical-validation}

\parhead{Conflicts of transactions}
Following the analysis of transaction conflicts by~\cite{miller2024parallel}, such conflicts can be fundamentally attributed to the storage state. They can be categorized into two types: conflicts over smart contract storage slots and conflicts over account balances. Other sources of conflict, such as contract creation and destruction, are less common in practice.

We therefore detect a conflict if two transactions either write to the same storage (a contract storage slot or an account balance), or if one reads from a storage slot while the other writes to it.
This definition naturally captures conflicts arising from MEV opportunities. For instance, in the case of arbitrage, two arbitrage transactions targeting a Uniswap V2 pool will both update the \texttt{reserve0} and \texttt{reserve1} storage slots of the pair contract, resulting in a write-write conflict.

\parhead{Dataset}
We randomly sampled $10{,}000$ Ethereum blocks from February 2025 to curate the dataset presented below. For each block, we collected publicly available transactions from Mempool Dumpster~\cite{flashbots2025mempoolDumpster} and private bundles submitted to Flashbots during the corresponding time window. The resulting dataset consists of $2.9\,\text{M}$ public transactions and $56.0\,\text{M}$ private bundles, amounting to $36.6\,\text{GB}$ of data.

Our conflict analysis is based on EVM storage read/write interactions, with the data from Flashbots providing a substantial sample of real private order flow for empirical evaluation.

\parhead{Empirical study}
We extend \texttt{rbuilder}~\cite{flashbots2025rbuilder} to analyze transaction conflicts within a given block. Specifically, for each block, we first simulate the execution of each order (where an order refers to either a single transaction or a bundle) and record the storage accessed during execution. We group transactions and bundles according to their storage conflicts: two transactions are placed in the same group if they conflict on storage. 
This allows us to measure both the number of conflict groups per block, denoted \numCG, and the number of orders within a conflict group, denoted \sizeCG.

The corresponding pseudocode for the grouping process is provided in~\autoref{alg:default} (\texttt{GetConflictGroups}).
We construct a storage-conflict graph: each order is represented as a node, and an edge connects two orders if they access a common storage slot during execution (L\ref{line:graphbegin}-\ref{line:graphend}).
The connected components of this graph define the conflict groups.
To efficiently identify conflicts, we maintain an index map $I$ from each storage slot to the set of orders that access it (L\ref{line:mapdef}, L\ref{line:mapupdate}), which allows us to locate conflicting orders without pairwise comparison.
With this index and a standard connected-component procedure (L\ref{line:connected}), we can obtain the conflict groups with low overhead.

\begin{figure}[th]
    \centering
    \includegraphics[width=0.95\linewidth]{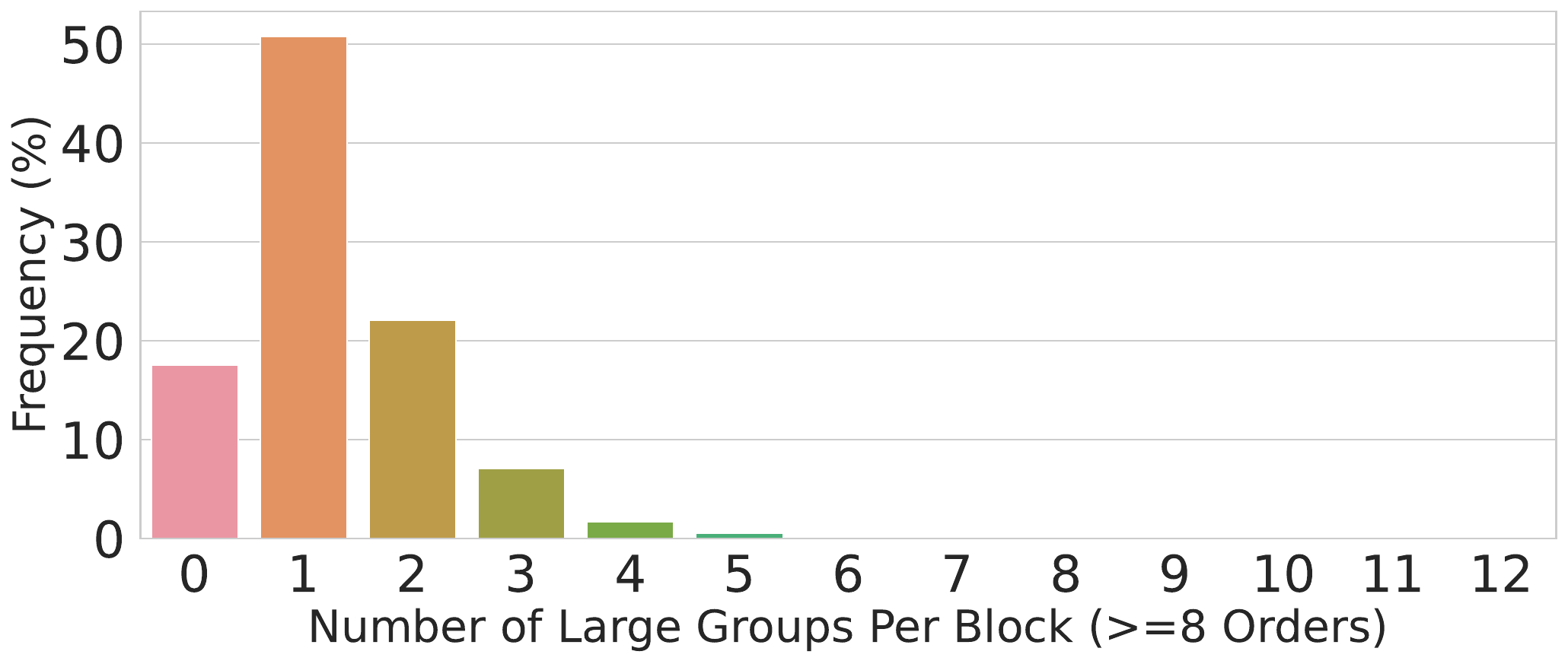}
    \caption{Frequency of large CGs where $\sizeCG \geq 8$.}
    \label{fig:large-groups}
\end{figure}

We examine how frequently large conflict groups show up when building blocks (i.e., groups with at least 8 orders, where permutation becomes less computationally efficient in our experiments). As shown in~\autoref{fig:large-groups}, in most blocks the number of large conflict groups is also relatively small: over 50\% of the blocks contain only one large conflict group, and in the majority of cases the number is fewer than four.
This finding is also consistent with practice, since most blocks contain only limited MEV opportunities.

\parhead{Insights}
We can summarize two insights from the empirical study above.
First, for most conflict groups, we can enumerate all possible orderings within the group to identify the one that maximizes the value. This is computationally feasible because the number of orders in each group is relatively small.
Second, for the remaining conflict groups, we can reduce complexity by exploiting their structural properties.
For example, if the conflicts arise purely from ERC-20 transfers, the ordering does not affect the outcome, so any ordering suffices. If the conflicts involve competition over an MEV opportunity, we only need to select the single most profitable transaction, since such an opportunity can be captured by only one transaction.
We incorporate these structural optimizations into our algorithm to reduce complexity in practice.

\subsection{A Concrete Proposal}
\label{sec:alg-evaluation}

\newcommand{\ncutoff}{k_\text{cutoff}}

\begin{algorithm}[htbp]
\footnotesize
\caption{Our Baseline Algorithm in \mechanism}
\label{alg:default}
\KwIn{Available orders $M$ and their bids $\mathbf{b}$}
\KwOut{Blocks $o^*$ and $\{o_{-i}\mid i\in M\}$}

\SetKwProg{Fn}{Function}{:}{}
\SetKwFunction{Access}{Access}
\SetKwFunction{ConnectedComponents}{ConnectedComponents}
\SetKwFunction{GetConflictGroups}{GetConflictGroups}
\SetKwFunction{Conflict}{Conflict}
\SetKwFunction{Bid}{Bid}
\SetKwFunction{Sample}{Select}
\SetKwFunction{Merge}{Merge}
\SetKwFunction{IsFeasible}{IsFeasible}
\SetKwFunction{AlgoFor}{AlgFor}

\BlankLine
\Fn{GetConflictGroups($M$)}{
    $G \gets (\mathcal{V}\!\leftarrow\!\emptyset,\ \mathcal{E}\!\leftarrow\!\emptyset)$; \tcp{initialize graph}
    $\mathcal{I} \gets \emptyset$; \tcp{map: storage slot $\to$ set of orders} \label{line:mapdef}
    \For{$o \in M$}{ \label{line:graphbegin}
        add $o$ to $\mathcal{V}$\;
        \For{$k \in \Access(o)$}{
            \For{$u \in \mathcal{I}[k]$}{ add edge $(o,u)$ to $\mathcal{E}$\; } \label{line:graphend}
            insert $o$ into $\mathcal{I}[k]$\; \label{line:mapupdate}
        }
    }
     $\{\mathsf{CG}_1,\dots,\mathsf{CG}_m\} \gets \ConnectedComponents(G)$; \label{line:connected}
     
\Return{$\{\mathsf{CG}_1,\dots,\mathsf{CG}_m\}$};
}

\BlankLine
\Fn{BlockBuilding($M$, $\mathbf{b}$)}{
    $\{\mathsf{CG}_1,\dots,\mathsf{CG}_m\} \gets \GetConflictGroups(M)$\;
    $\mathcal{Q} \gets [\,\mathsf{CG}_1,\dots,\mathsf{CG}_m\,]$\; %
    $\mathcal{B} \gets [\,]$\;
    
    \While{$|\mathcal{Q}| > 0$}{ \label{line:begin}
      $CG \gets$ pop from $\mathcal{Q}$\;
    
      \uIf{$|CG| < \ncutoff$}{
        $o^\star \gets \arg\max_{o \in \Omega_{CG} } \Bid(o)$\; \label{line:best}
      }
      \uElseIf{\IsFeasible$(CG)$}{ \label{line:is-feasible}
        $o^\star \gets \AlgoFor(CG)$; \label{line:algfor}
        }
      \Else{
        $CG' \gets \Sample(CG,\ncutoff-1)$; \label{line:select}
        $o^\star \gets \arg\max_{o \in \Omega_{CG'} } \Bid(o)$\;
        }
      $\mathcal{B} \gets \mathcal{B} \,\|\,  o^\star$\; \label{line:end}
    }
\BlankLine
\Return{$\mathcal{B}$};
}

$o^* \gets BlockBuilding(M, \mathbf{b})$ \;
$o' \gets [\,]$\;
\ForAll{$i \in M$}{
   $o_{-i} \gets BlockBuilding(M, \left(\mathbf{b}_{-i}, b_i=0 \right))$\;%
  append $o_{-i}$ to $o'$\;
}

\Return{$o^*, o'$};

\BlankLine
\textbf{Auxiliary definitions:}\\
$\Omega_{CG}$: the set of all permutations of the orders in $CG$\;
\Access($o$): the storage slots accessed by $o$\;
\ConnectedComponents($G$): return all connected components in the graph $G$\;
\Bid{$o$}: the total bid of the orders in $o$\;
\IsFeasible{$CG$}: true if the complexity can be simplified\;
\AlgoFor{$CG$}: corresponding simplified solution\;
\Sample{$CG, k$}: select $k$ orders from $CG$.
\end{algorithm}

\autoref{alg:default} presents our concrete proposal, following the above design principles. 
It first partitions these orders into conflict groups, where two orders belong to the same group if they conflict with each other.
Each conflict group is then resolved independently (L\ref{line:begin}-\ref{line:end}).
For groups of less than $\ncutoff$ transactions, the algorithm enumerates all subsets and selects the ordering with the highest value (L\ref{line:best}).
We choose $\ncutoff$ such that exhaustive permutation of $\ncutoff$ transactions remains fast enough.
For groups larger than $\ncutoff$, exhaustive search becomes infeasible due to the larger search space, so the algorithm applies alternative resolution strategies. 

For large conflict groups, we first check whether they fall into several special cases (L\ref{line:is-feasible}-\ref{line:algfor}).
E.g., if the group consists of bundles all involving the same transaction (e.g., every sandwich attack targets the same victim, or every backrunning arbitrage targets the same opportunity creator), only one bundle can be included. In this case, the complexity reduces to linear time, since we only need to compare all candidate bundles and select the best one.
If all transactions in the group are directed to the same contract address, their relative order is usually irrelevant. We therefore simply apply a random ordering using an internal seed.
As a last resort, if the group does not match any simplified case, we deterministically order the transactions by their hashes using an internal seed and select the top $\ncutoff-1$ (L\ref{line:select}) for exhaustive permutation.

At the end, the results from all conflict groups are merged to form the set of orders that will be included in the block.

\parhead{Baseline algorithm properties}
The baseline algorithm satisfies two conditions of~\autoref{property} required in~\autoref{subsec:practical}.
First, the feasible set of blocks considered by our baseline algorithm is determined independently of the bids. The algorithm begins by partitioning orders based on their storage conflicts, which does not depend on the bids. Each conflict group is then handled separately in three cases: (1) if the group size is less than $\ncutoff$, \textit{all} possible permutations of orders in the conflict group are treated as the feasible set; (2) if a conflict group is large but falls into a simplified structure, \textit{all} corresponding blocks are considered; and (3) otherwise, a subset is deterministically selected without reference to bids. Thus, it is clear that the feasible set of blocks under our baseline algorithm is independent of the bids.

Second, the complexity of conflict groups remains manageable. Our baseline algorithm partitions transactions into conflict groups and handles each group as follows: (1) if the group size is relatively small, we can enumerate all permutations within a short time; (2) if the group size is large but has structural properties that allow simplification, we can identify the optimal order in linear time; (3) if no ordering is required, the process introduces no additional complexity; and (4) if none of these conditions hold, we deterministically select seven orders from the group so that permutation remains computationally feasible.
Note that in the final step, merging all conflict groups is linear since no ordering needs to be resolved. Hence, the overall process is computationally feasible.

\subsection{Block-building Performance}

\parhead{Implementation}
We implement \autoref{alg:default} (\texttt{BlockBuilding}) in about 900 lines of Rust code and 83 lines of Python code. Our implementation is built on top of \texttt{rbuilder}, which we also use to analyze transaction conflicts. We set $\ncutoff = 8$ based on multiple experiments, which show that the runtime increases significantly when $\ncutoff > 8$ (see \autoref{app:cutoff-setup} for details).\done%

\parhead{Testbed} Our evaluation is conducted on a server running Ubuntu 24.04 with 60 CPU cores and 128GB of RAM.

\parhead{Evaluation setup}
For a given block height, we initialize an independent simulation environment that replays the block-building context at that height.
The inputs consist of public mempool transactions and private bundles sourced as described in~\autoref{sec:empirical-validation}.
Private bundles are provided to all algorithms as atomic units: each bundle must either be included with its internal transaction order preserved or excluded entirely.

Within this environment, we compare the block value output by \autoref{alg:default} against those produced by algorithms implemented in \texttt{rbuilder}~\cite{flashbots2025rbuilder}, including two greedy block-building algorithms and the parallel block-building algorithm.
We repeat this evaluation over $10{,}000$ Ethereum block heights that are disjoint from those used in our empirical study, providing an independent benchmark.
For each block height, we compute the block value achieved by each block-building algorithm and check whether the block constructed by our baseline algorithm has the highest value among the compared algorithms.

\parhead{Results}
We find that in \betterratio of cases, the baseline algorithm produces the ``optimal'' block, defined as the highest-value block among the evaluated algorithms.
Across all non-optimal cases, the median absolute difference between the baseline algorithm and the best algorithm is 0.0033 ETH, while the median relative difference is 8.6\%. Although this gap is small, it is still important to understand why \autoref{alg:default} does not win.

\begin{figure}
    \centering
    \includegraphics[width=\linewidth]{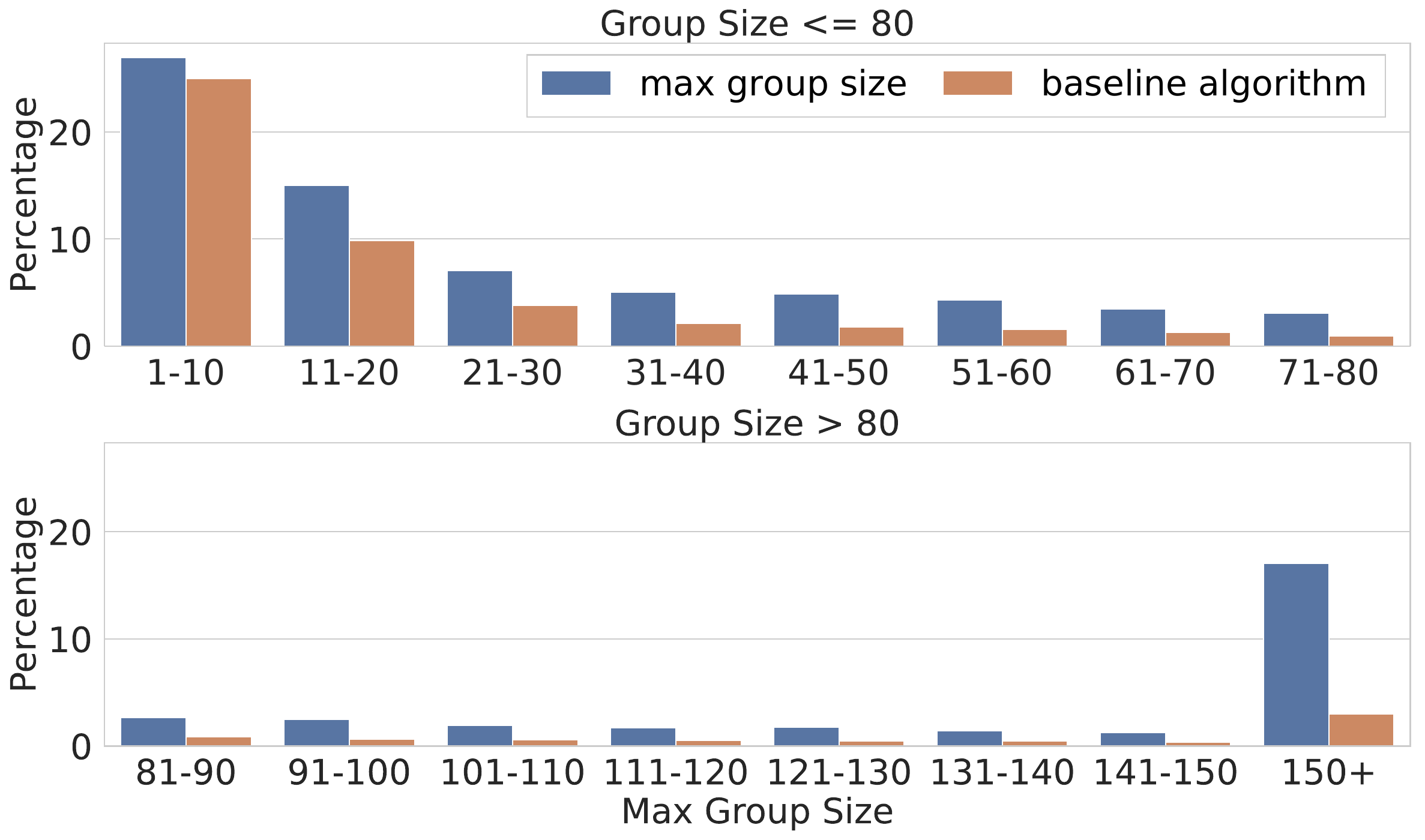}
    \caption{Distribution of the maximum group size in each block, and the fraction of all blocks where the baseline algorithm is the best in the block with that maximum group size.}
    \label{fig:group-analysis}
\end{figure}

Our analysis shows that the main limitation comes from large conflict groups.
As shown in~\autoref{fig:group-analysis}, when the maximum group size in a block is at most 10, the baseline algorithm finds the optimal block in most cases.
However, its performance declines as the maximum group size increases.

When a large group cannot be simplified by our special-case optimization, the baseline algorithm keeps permutation feasible by selecting only a subset of transactions from the group.
Many excluded transactions could still be validly included, so greedy algorithms may build higher-value blocks by including more transactions, even without guaranteeing the optimal ordering.
This effect is especially clear for high-value blocks: among blocks with value above 1 ETH, the baseline algorithm is optimal in only about 5\% of cases.
Thus, when MEV is high, its feasibility-oriented permutation strategy can sacrifice profitable bundles and reduce block value.
We return to possible ways to improve this tradeoff in~\autoref{app:improvement-ideas}.

\parhead{Execution time evaluation}
Beyond block value, we also measure execution time, since practical block building must fit within Ethereum's slot-timing constraints. Counterfactual computations for searcher refunds can run asynchronously, so our timing evaluation focuses on candidate-block construction.

The results are shown in~\autoref{fig:time-cdf}.
Because the baseline algorithm permutes candidate sequences at runtime, it can naturally leverage multi-threading; we therefore evaluate it with 25 and 50 threads.
As shown in~\autoref{fig:time-cdf}, the two greedy algorithms (\texttt{mgp-ordering} and \texttt{mp-ordering}) are the fastest, followed by the parallel algorithm.
For the baseline algorithm, multi-threading substantially reduces runtime; with 50 threads, it approaches the parallel algorithm in most cases, supporting its practicality under Ethereum's slot timing.

\begin{figure}
    \centering
    \includegraphics[width=\linewidth]{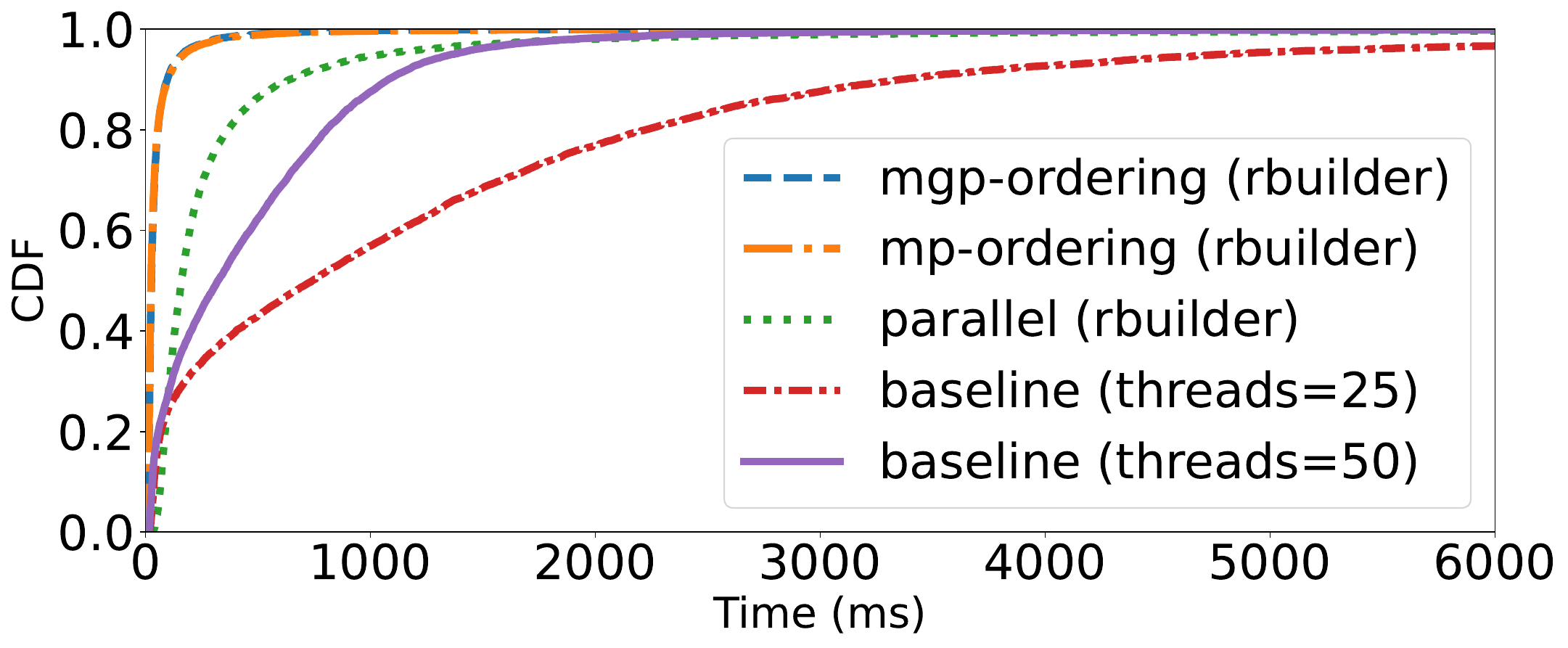}
    \caption{CDF of runtime for each block-building algorithm.}
    \label{fig:time-cdf}
\end{figure}

%% file: sections/6-discussion.tex
\section{Discussion}
\done%

\parhead{Rethinking block-building mechanism design}
This paper highlights a central lesson in block-building mechanism design: while eliminating incentives toward integration and ensuring incentive compatibility are conceptually clean and desirable goals, they come into inherent tension with the practical objective of economically efficient block construction.

Looking back, the VCG-based mechanism $\texttt{VCG}_\mathcal{A}$ presented in~\autoref{subsec:practical}, where the block is constructed solely by a baseline algorithm, %
can serve as a ``perfect'' mechanism in the following sense. First, it inherently eliminates the integration issue, as there are no external builders with whom searchers can integrate. 
Second, truthful bidding is a dominant strategy for searchers under $\texttt{VCG}_\mathcal{A}$ (\autoref{theorem:variantDSIC}). The primary concern of such a mechanism is its economic efficiency. The baseline algorithm may fail to generate a block that efficiently utilizes the block space and maximizes the total realized value (i.e., social welfare). To overcome this, there are two possible directions.

The first direction is to design a better baseline algorithm. However, the block-building task fundamentally involves ordering and interactions among transactions, and thus appears to subsume difficult combinatorial optimization problems. This suggests that even achieving theoretically good approximation guarantees is challenging (under standard complexity assumptions such as $\mathsf{P} \neq \mathsf{NP}$). On the other hand, one might hope to design a baseline algorithm with strong approximation guarantees that is ``good enough'' \emph{in practice}, which is the direction we explore in~\autoref{sec:implementation}. In particular, we believe it is worthwhile to further explore the design of practically effective approximate algorithms based on empirical insights. %

The second direction is to allow builders to participate, as we do in this paper, with the hope of leveraging their sophistication to generate higher-value blocks. 
However, this reintroduces the searcher-builder integration problem and complicates the incentive landscape. 
Even setting aside the integration concerns and assuming that builder-submitted algorithms also satisfy the two conditions in~\autoref{property} (just as the baseline algorithm does), 
we conjecture that it is impossible to design a mechanism that simultaneously satisfies \textit{Efficiency} (always selecting the outcome that maximizes social welfare), \textit{Searcher-DSIC}, \textit{Builder-DSIC}, \textit{IR}, and \textit{Budget Balance}. 
Our reasoning behind this conjecture is discussed in~\autoref{sec:conjecture-intuition}.
\done%

\parhead{Sybil attacks}
Sybil attacks, where a builder or searcher poses as multiple actors, are a significant challenge in MEV-related mechanisms~\cite{mazorra2023cost,mazorra2023optimality,pan2024sybil,todo2013false}.
We discuss their implications here.

First, when builders create multiple identities to participate in the competition, they cannot increase their overall profit.
The winning builder receives only the difference between its block value and the second-best block value.
Creating additional identities cannot lower the value of the second-best block and therefore cannot increase the builder's profit.

Sybil attacks by searchers are more detrimental. 
Recall that the mechanism computes refunds based on each bundle's marginal contribution. 
A searcher can split a single bundle into multiple mutually compatible Sybil bundles to inflate their total marginal contributions, thereby extracting a higher refund compared to the case without Sybil. 
Note that MEV-Boost also faces searcher Sybil issues, but searchers are not explicit participants in the mechanism, so such behavior is neither modeled nor addressed at the mechanism level.

The existing line of work on mitigating Sybil attacks mainly includes the introduction of decentralized identities for searchers~\cite{maram2021candid} and the application of Sybil-proof mechanisms~\cite{pan2024sybil}.
Another mitigation approach is to introduce a reputation system and classify searchers into high- and low-reputation queues based on historical submissions~\cite{flashbots2025searcherreputation}.
When the baseline algorithm reduces its search space, it could prioritize bundles from high-reputation searchers, while still preserving bid-independent selection. This may discourage the use of new Sybil identities that cannot inherit an established reputation.
We leave the detailed design to future work.

\parhead{Proposer adoption} 
Our mechanism is mainly designed to address the incentives of searchers and builders. 
While the proposer's incentive for adoption is not the focus of our paper, we briefly explore it in~\autoref{sec:adoption}.
At a high level, we model the interaction between a proposer and searchers as a game in which the proposer chooses whether to commit to \name or to build blocks directly and pick the most profitable outcome. We find that in two natural extremes---when bundles are either fully independent or fully conflicting---the equilibrium is that the proposer commits to adopting \name.

These results suggest that, at least in certain scenarios, proposer adoption can be self-enforcing. Extending this analysis to more general settings remains an intriguing direction for future work. The formal model and analysis appear in~\autoref{sec:adoption}.

%% file: sections/7-related-works.tex
\section{Related Work}

\parhead{Builder centralization}
Early empirical work quantifies the concentration of Ethereum's builder market~\cite{yang2022sok,heimbach2023ethereum,wahrstatter2023time}, while later work studies its causes and implications, linking concentration to exclusive private order flows (``integration''), which advantage integrated builders, reduce proposer profits, and may threaten the stability of PBS~\cite{gupta2023centralizing,yang2024decentralization,oz2024wins}. Theoretical work explains this trend through mechanism-level incentives: exclusive order flows make integrated builders more likely to win MEV-Boost auctions and attract private transactions from non-integrated searchers~\cite{gupta2023centralizing,capponi2024proposer}. Meanwhile, competitive builder markets can support validator decentralization, whereas heterogeneous block-building capabilities can drive validator centralization~\cite{bahrani2024centralization}. These findings motivate our goal of eliminating exclusive order flows and ensuring open builder access, while still rewarding algorithmic innovation.

\parhead{Block-building in TEEs}
TEE-based block building is increasingly used across blockchain systems. On Ethereum, BuilderNet is the most widely deployed example and accounts for about 25\% of the market share at the time of writing~\cite{relayscan}, but it lacks a formal mechanism for incentive alignment and block-building algorithm innovation. Beyond Ethereum, Solana's Blockspace Assembly Markets (BAM) use TEEs for transaction filtering and ordering~\cite{bam2025introducing}, while Flashbots' Rollup-Boost, adopted by Unichain and Base, runs the mempool and block builder inside TEEs for transaction confidentiality, integrity, and TEE-attested block production~\cite{BaseFlashblocksApps2025,flashbots2024rollupboost}, with planned support for multi-operator block building~\cite{flashbots2024rollupboost}. \name can complement these systems by encouraging truthful behavior from searchers and builders while supporting competition among block-building algorithms.

\parhead{Theoretical analysis of MEV}
Mamageishvili et al. study MEV allocation between validators and searchers in competitive block-building markets through the core of a cooperative game, showing that marginal-contribution-based allocations are core-stable and dominant-strategy incentive compatible~\cite{mamageishvili2024searcher}. In contrast, we include builders and study MEV allocation among searchers, builders, and validators. Other theoretical works analyze MEV in CFMMs~\cite{kulkarni2022towards,zhang2024rediswap,wadhwa2024data,li2026trilemma,shi2026otter}, batch auctions~\cite{zhang2025maximal} and sequencing rules~\cite{xavier2023credible,li2023mev}.

%% file: sections/8-conclusions.tex
\section{Conclusions}

In this paper, we propose \name, a framework for equitable and incentive-compatible block building that decouples transaction collection from ordering to ensure equal access to order flows and foster builder innovation.
At its core, \mechanism with a baseline algorithm incentivizes truthful participation and innovation in block-building algorithms.
We further implement a baseline algorithm guided by empirical conflict analysis, which achieves the highest block value in \betterratio of cases, demonstrating practicality and effectiveness.

%% file: sections/ethics.tex
\section*{Ethics Considerations}

This work designs and evaluates a mechanism for decentralizing Ethereum block building. 
Our evaluation does not involve any live-chain experiments, does not submit transactions, and does not interact with production blockchain infrastructure. 
All experiments are conducted offline using historical blockchain states and archival datasets.

Our use of transaction and bundle-related data is limited to aggregate measurements needed to evaluate the mechanism. 
We do not analyze, profile, or disclose individual users, searchers, builders, addresses, or specific bundle strategies. 
All reported results are aggregated and historical, limiting privacy, reputational, and operational risks to ecosystem participants.

%% file: sections/appendix.tex
\section{Omitted Proofs}
\label{sec:omitted-proofs}

\subsection{Proof of~\autoref{theorem:variantDSIC}}\label{app:proof-variantDSIC}

\begin{proof}
    Condition (ii) enables the mechanism $\texttt{VCG}_\mathcal{A}$, especially the step (1), to run efficiently.

    Fix a bundle set $M$, an arbitrary bundle $i$, and others' bids $\mathbf{b}_{-i}$. Given any $b_i$, suppose the chosen best block is $o^*\in \Omega_\mathcal{A}(M,\mathbf{b})$. Then $i$'s utility is: $$v_i(o^*) - b_i(o^*) + r_i = \underbrace{\left[ v_i(o^*) + \sum_{j \ne i} b_j(o^*) \right]}_{(\text{A})} -\underbrace{\left[ \max_{o \in \Omega_\mathcal{A}} \sum_{j \ne i} b_j(o) \right]}_{(\text{B})}.$$ 

    First, by condition (i), for any $i$'s bid $b_i$, $b_i'$, we have $\Omega_\mathcal{A}(M, \mathbf{b}_{-i}, b_i) = \Omega_\mathcal{A}(M, \mathbf{b}_{-i}, b_i')\triangleq \Omega_\mathcal{A}$, thus $\Omega_\mathcal{A}$ and the entire term $(B)$ does not depend on $i$'s bid. 
    Then, to maximize utility, bundle $i$ should aim to maximize the first term (A). If bundle $i$ sets $b_i = v_i$, the term (A) becomes $b_i(o^*) + \sum_{j \ne i} b_j(o^*) = \sum_{j\in M} b_j(o^*)$, a function of the selected block, which is exactly the algorithm $\mathcal{A}$'s consideration when choosing $o^*$ and the objective that the mechanism maximizes in selecting the final block. 
    Thus, truthful reporting aligns $i$'s incentive with the mechanism's objective and leads to a utility-maximizing outcome for $i$. %
\end{proof}

\subsection{Proof of~\autoref{claim:non-negative-refund}}\label{app:proof-non-negative-refund}
\begin{proof}
    For any bundle $i \in S$, the refund equals their bid, which is non-negative by definition. 
    For any $i \in M\setminus S$, the refund is $r_i(\mathbf{b}) = \sum_{j\in M\setminus S} b_j(o^*) - \sum_{j\neq i} b_j(o_{-i})$. Since $o^*$ maximizes the total bid among all possible outcomes in $\Omega_\mathcal{A}$, we have $$\sum_{j\in M\setminus S} b_j(o^*) \geq \sum_{j\in M\setminus S} b_j(o_{-i}) \geq \sum_{j\neq i} b_j(o_{-i}).$$ 
    It follows that $r_i(\mathbf{b}) \geq 0$.
\end{proof}

\subsection{Proof of~\autoref{thm:proposerIR}}\label{app:proof-proposerIR}

\begin{proof}
    Fix any $i\in M\setminus S$. 
    Let $o_{-i} \in \arg \max_{o\in \Omega_\mathcal{A}} \sum_{j\neq i} b_j(o)$. By the definition of $o_{-i}$, $\sum_{j\neq i} b_j(o_{-i}) \geq \sum_{j\neq i} b_j(o^*)$. Thus, 
    \begin{align*}
        r_i(\mathbf{b}) &= \sum_{j\in M\setminus S} b_j(o^*) - \sum_{j\neq i} b_j(o_{-i}) \\
        &\leq \sum_{j\in M\setminus S} b_j(o^*) - \sum_{j\neq i} b_j(o^*) = b_i(o^*).
    \end{align*}        
    It follows that total refunds to the set $M\setminus S$ of searchers satisfy
    $$\sum_{i\in M\setminus S} r_i(\mathbf{b}) \leq \sum_{i\in M\setminus S} b_i(o^*) = \beta_0.$$

    If the baseline algorithm wins, builders receive no refund, so the total refund is $\sum_{i\in M\setminus S} r_i(\mathbf{b}) + \sum_{i\in S} b_i(o^*) \leq \beta_0 + \sum_{i\in S} b_i(o^*)$, as claimed. 
    Otherwise, when a builder wins, the total refunds to searchers and builders equal 
    $$\sum_{i\in M\setminus S} r_i(\mathbf{b}) + \sum_{i\in S} b_i(B^*) + \beta^* - \max\{\beta_0, \beta'\} \leq \beta^* + \sum_{i\in S} b_i(B^*),$$ 
    which is the payment from the winning builder to \mechanism. 
    
    This concludes the proof.
\end{proof}

\subsection{Proof of \autoref{theorem:conflictfreeIntegration}}
\label{app:proof-conflictfreeIntegration}

Before the proof, it is worth noting that integration couples the bundle $i$'s inclusion with the integrated builder's algorithm, so the integrating searcher $i$ and the builder $j$ should be viewed as a single coalition. 
Although the builder itself faces a second-price auction, integration may incentivize her to misreport if doing so could increase their joint utility. 
Therefore, our analysis also considers builder $j$'s strategic behavior (i.e., reporting arbitrarily), and examines whether any deviation involving integration and/or misreporting can improve the pair's joint utility relative to their desired behavior (i.e., non-integration and both truthful).

\begin{proof}
    Fix a set $M$ of bundles, a set $S \subseteq M$ of conflict-free bundles, and $n$ builder-submitted algorithms. Let $i \in S$ be an arbitrary conflict-free bundle, and let $j \in [n]$ be a builder who may potentially integrate with $i$.

    Suppose that searcher $i$ participates (i.e., does not integrate) and bids truthfully, while builder $j$ also bids truthfully, which we refer to as their \emph{desired behavior}. 
    Given $M$, $(\mathbf{v}_i, \mathbf{b}_{-i})$, and $n$ building algorithms, the mechanism runs all algorithms on $M\setminus S$ and determines the winner according to their bids $\set{\beta_0, \beta_1, \cdots, \beta_n}$. 
    Let $V_j(M\setminus S)$ denote the value of $j$'s generated block, and let $\hat{\beta} = \max \set{\beta_k \mid k\in [0:n], k \neq j }$ be the highest competing value from all other algorithms. 
    Note that both $V_j(M\setminus S)$ and $\hat{\beta}$ are independent of searcher $i$'s strategy (and builder $j$'s bid), since bundle $i$ is excluded from the input of all algorithms.

    Next, we consider two exhaustive cases and show that in both cases, any possible deviation by the searcher $i$ and the integrated builder $j$ yields weakly lower total utility compared to their desired behavior. 

    \paragraph{Case 1: $j$ is the winner.} 
    This corresponds to $V_j(M\setminus S) \geq \hat{\beta}$. 
    In this case, the joint utility of $i$ and $j$ under their desired behavior (i.e., non-integration and both truthful) is:
    $$U_1 = v_i + V_j(M\setminus S) - \hat{\beta},$$
    where $v_i$ is $i$'s true valuation. 

    Consider all possible deviations by searcher $i$ and builder $j$. As long as builder $j$ still wins, bundle $i$ will be included in the final block regardless of whether $i$ integrates or not, and their total utility remains $U_1$. On the contrary, if builder $j$ loses,
    \begin{itemize}
        \item If searcher $i$ integrates, the bundle is excluded from inclusion and their joint utility becomes $0 < U_1$.
        \item If searcher $i$ participates, the bundle remains included by the winner, yielding total utility $v_i \leq U_1$.
    \end{itemize}

    \paragraph{Case 2: $j$ is not the winner.} 
    This corresponds to $V_j(M\setminus S) \leq \hat{\beta}$.
    In this case, bundle $i$ is still included in the winning block and receives the full refund $v_i$. The joint utility of searcher $i$ and builder $j$ under their desired behavior is thus: 
    $$U_2 = \underbrace{v_i}_{\text{$i$'s true valuation}} 
    - \underbrace{v_i}_{\text{$i$'s payment}} 
    +  \underbrace{v_i}_{\text{refund}} 
    = v_i.$$
    
    We again consider all possible deviations. If the deviation causes builder $j$ to win, regardless of searcher $i$'s strategy, their total utility becomes $v_i + V_j(M\setminus S) - \hat{\beta} \leq v_i = U_2$. On the contrary, if builder $j$ still loses,
    \begin{itemize}
        \item If searcher $i$ integrates, their joint utility will be $0 < U_2$.
        \item If searcher $i$ participates, the bundle remains included and the joint utility stays $v_i = U_2$.
    \end{itemize}

    In both cases, every deviation yields the same or strictly lower utility. Hence, the searcher's strategy pair, \emph{(participate, bid truthfully)}, weakly dominates all alternatives with respect to the joint utility of $i$ and $j$. 
    This concludes the proof. 
\end{proof}

\section{Additional Empirical Evidence}

\subsection{Ethereum Gas Usage}
\label{app:gas-usage}

\begin{figure}[htbp]
    \centering
    \includegraphics[width=\linewidth]{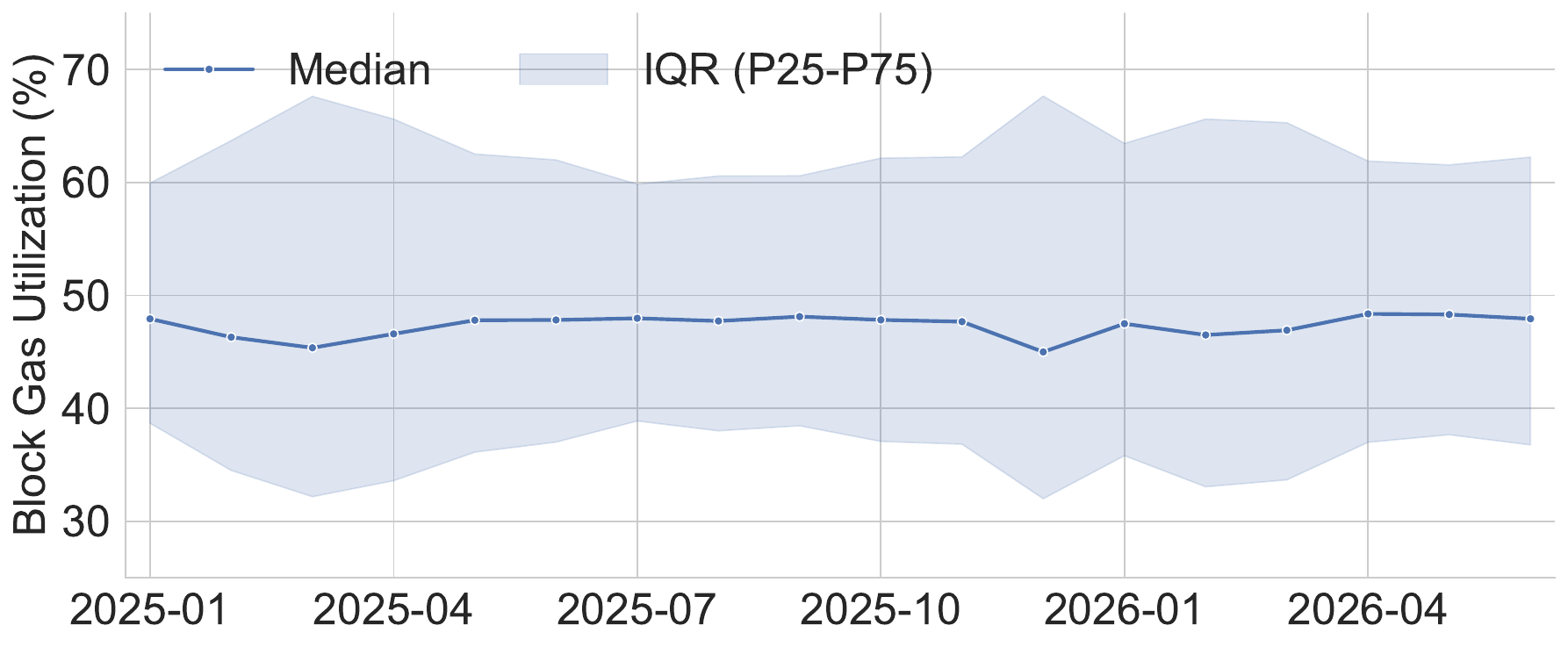}
    \caption{Monthly block gas utilization on Ethereum. The line shows the median, and the shaded area represents the interquartile range.}
    \label{fig:gas-usage}
\end{figure}

To examine blockspace utilization under EIP-1559, we measure the gas utilization of each Ethereum block as the ratio of gas used to its gas limit from January 2025 to June 2026. As shown in \autoref{fig:gas-usage}, the monthly median utilization remains between 45.0\% and 48.4\%, while the 75th percentile remains below 68\%. These results suggest that EIP-1559 works as intended, leaving substantial unused gas capacity in most blocks and supporting our assumption that all bundles in $M$ can fit within a single block.

\subsection{Banana Gun Transaction Conflicts}
\label{app:empirical-study-of-banana-gun}

We conduct an empirical study to measure how often Banana Gun transactions conflict with other transactions within the same block. Specifically, we randomly select $10{,}000$ on-chain blocks and group transactions according to their conflicts. For each conflict group containing Banana Gun transactions, we check whether it also includes any non-Banana Gun transactions.

In about $37\%$ of these blocks, all conflict groups that include Banana~Gun transactions contain no other transactions, which aligns with our definition of ``conflict-free.'' In particular, one example is block~21000813, which has a block value of $55.5$~ETH, where the conflict-free Banana~Gun transactions contribute the majority of the block value.
This observation further justifies our assumption that Banana~Gun bundles can be treated as conflict-free in a significant portion of blocks.

\section{An alternative refund rule: seemingly natural but vulnerable to collusion}\label{sec:alternative-refund}
In~\mechanism, the refund $r_i$ to each bundle $i\in M\setminus S$ is computed solely from the baseline algorithm, while the refund to each bundle $i\in S$ is a fixed value independent of ordering. That means in either case, the refunds to searchers are identical, regardless of which block is ultimately selected. 
One might wonder why we do not always use the \emph{winning} algorithm to calculate the refund. Namely, when a builder-submitted algorithm wins (Case 4b), define the refund to each $i\in M\setminus S$ instead as 
\begin{equation}\label{eq:alternative}
r_i' =
\begin{cases}
w_i, & \text{if } \sum_{j\in M\setminus S} w_j \leq \max\{\beta_0, \beta'\}; \\
f(w_i),  & \text{otherwise},
\end{cases}
\end{equation} 
where $w_i = \beta^* - \beta_{-i}$, with $\beta_{-i}$ denoting the winning builder's counterfactual bid if $b_i$ were always $0$. Here, $f(w_i)$ is any function satisfying that $f(w_i)\geq 0$ for all $w_i$ and $\sum_{i\in M\setminus S} f(w_i) \leq \max\{\beta_0, \beta'\}$. 
While seemingly natural, the modified mechanism with this alternative refund is vulnerable: a strategic builder could collude with a searcher to steal value from the proposer and other searchers, as illustrated below.

\begin{observation}\label{ob:alternative}
    For any instance $M$ and $\textbf{b}$ such that the baseline algorithm outperforms all builder-submitted algorithms, the alternative refund rule $r_i'$ above is not robust against strategic collusion between a builder and a searcher.
\end{observation}

\begin{proof}
    Fix the bundle set $M$, the bid profiles $\mathbf{b}$, the baseline algorithm, a colluding pair consisting of a builder $j$ and searcher $i$, and all algorithms from other builders $k\neq j$. %
    Under the instance $M$, the baseline algorithm produces $o^*$ with the value $\beta_0$ higher than all builder bids, including $j$’s. As a result, builder $j$ loses and gets zero utility, while searcher $i$'s utility is 
    $$u_i = v_i(o^*) - b_i(o^*) + \beta_0 - \sum_{k\neq i} b_k(o_{-i}),$$
    where $v_i(o^*)$ is $i$'s true valuation under the block $o^*$, $b_i(o^*)$ is $i$'s payment, and $\beta_0 - \sum_{k\neq i} b_k(o_{-i})$ is the refund to $i$. 
    
    Now, let builder $j$ adopt the following strategy: submit an algorithm that copies the baseline algorithm's code to generate the same block, bids $\beta_0+\epsilon$, and reports counterfactual bids $\beta_{-i} = \epsilon$ for the colluding bundle $i$ and $\beta_{-k} = \beta_0+\epsilon$ for all other bundles $k \neq i$. 
    
    By bidding $\epsilon$ above the baseline algorithm, the strategic builder $j$ instead wins. Thus, the mechanism enters Case (4b) with the alternative refund rule, which requires $j$'s algorithm to also return a vector of counterfactual bids $\{\beta_{-i} \mid i\in M\}$. 
    In this case, the winning block is still $o^*$, as the builder $j$ simply copies the baseline algorithm's block. The builder $j$'s utility remains $0$ (by getting $\beta_0$ from searchers, paying $\beta_0+\epsilon$ to the TEE, and receiving $\epsilon$ as refund per the second step of Case 4b). However, the colluding searcher $i$'s utility increases to 
\[
\begin{aligned}
u_i' & \triangleq v_i(o^*) - b_i(o^*) + (\beta^*-\beta_{-i}) \\
   &= v_i(o^*) - b_i(o^*) + (\beta_0 + \epsilon - \epsilon) \\
   &= v_i(o^*) - b_i(o^*) + \beta_0 > u_i \,,
\end{aligned}
\] gaining the maximum refund $\beta_0$. This concludes the proof.
\end{proof}

Intuitively, the value $\beta^*-\beta_{-i}$ signals the marginal importance of bundle $i$, which decides the refund. By collusion, a builder could make $i$'s bundle appear more critical than it truly is, by misreporting $\beta$'s.
As a result, $i$ (together with the colluding builder) receives an inflated refund at the expense of others.

\section{Additional Evaluation and Optimization Details}

\subsection{Cutoff Selection}
\label{app:cutoff-setup}

We choose the cutoff $\ncutoff$ empirically based on the baseline algorithm's running time.
Starting from a cutoff of 10, we gradually decrease the cutoff value and record the running time under each setting.
To make the comparison fair, all configurations in this experiment use the same parallelism level of 50 threads.
As shown in~\autoref{fig:cutoff-runtime}, a cutoff of 9 still leads to long running times in some cases, whereas a cutoff of 8 gives a more practical running-time profile.
Therefore, we use $\ncutoff=8$ throughout our evaluation.

\begin{figure}[htbp]
    \centering
    \includegraphics[width=\linewidth]{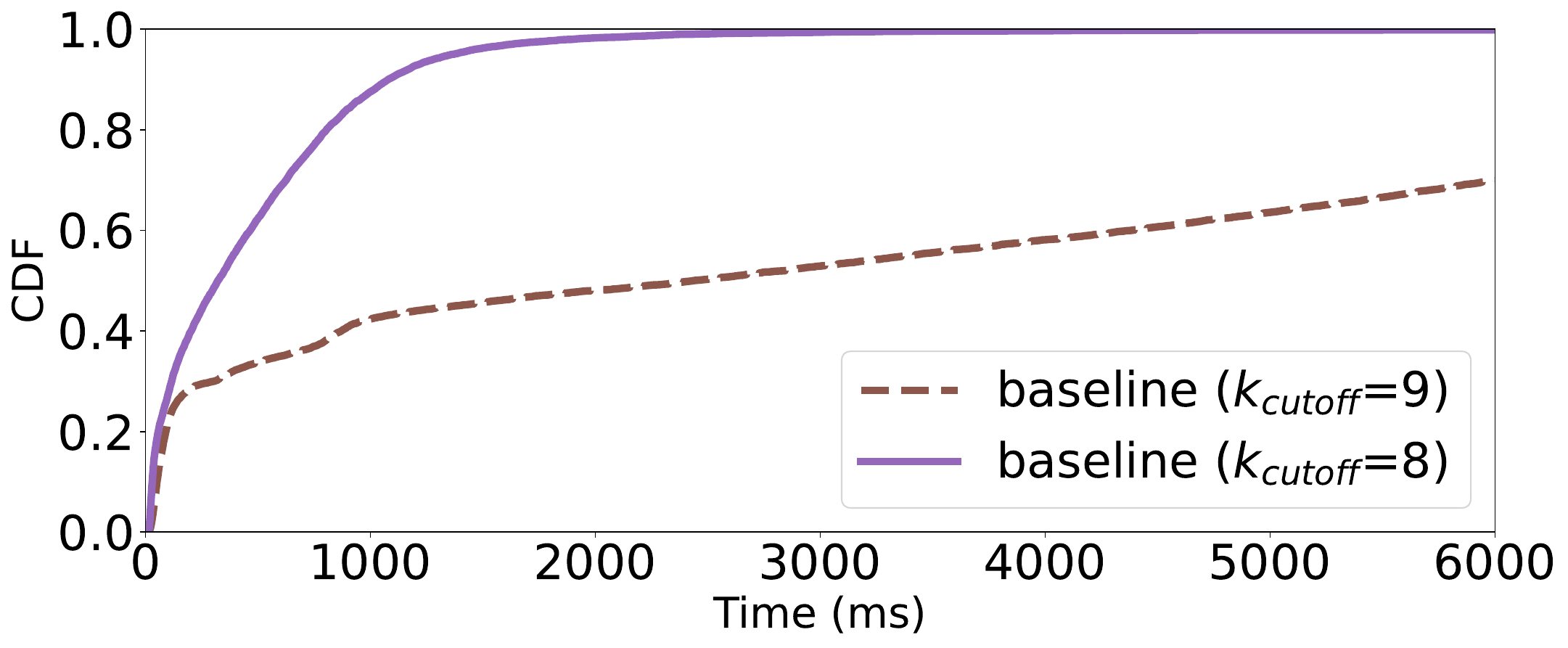}
    \caption{CDF of runtime for baseline algorithm with different $\ncutoff$.}
    \label{fig:cutoff-runtime}
\end{figure}

\subsection{Potential Optimizations}
\label{app:improvement-ideas}

\parhead{Directions for optimization}
The evaluation shows that although the baseline algorithm remains close to optimal in most non-optimal cases, there is still room for improvement. As shown in~\autoref{fig:group-analysis}, the baseline algorithm performs well when the maximum group size in a block is small, but its effectiveness drops sharply as the group size increases. In such cases, simply enumerating more possibilities is insufficient; instead, it is necessary to identify additional scenarios where the permutation complexity can be simplified (L\ref{line:is-feasible}–\ref{line:algfor}), thereby further improving the performance.
For example, if a more detailed study enables us to develop methods for resolving groups with sizes between 11 and 20, the winning rate of the baseline algorithm might increase from 53\% to 58\%.

Moreover, the runtime results suggest an orthogonal direction: improving how compute resources are coordinated could further reduce the execution time of the baseline algorithm. 
This includes better parallel scheduling across CPU cores and, where applicable, hardware acceleration for parts of the computation that are amenable to GPU execution.

\parhead{A preliminary attempt}
The above observations motivate optimization directions that exploit additional structure in large conflict groups.
As a preliminary test, we examine whether a general equivalence-class reduction can remove enough redundant permutations to make exhaustive search more practical.

Within a conflict group $CG$, not all permutations are distinct for our algorithm. 
Only the relative ordering of nodes connected by a conflict edge matters. 
Two nodes may belong to the same conflict group $CG$ even if they do not directly conflict, because they can both conflict with other nodes in the group. 
For such a non-conflicting pair, swapping their relative ordering does not change the result considered by the algorithm.
We therefore treat two permutations of $CG$ as equivalent if they agree on the relative order of every conflicting pair. 
This observation suggests a possible way to reduce the search space: instead of enumerating all $|CG|!$ permutations, we can enumerate only one representative ordering from each equivalence class. 
Equivalently, each representative corresponds to an acyclic orientation of the conflict graph induced by $CG$.

To assess whether this reduction is useful, we implemented a lightweight prototype based on \texttt{rbuilder}. 
The prototype enumerates one canonical ordering per acyclic orientation for conflict groups and compares the resulting search space with full permutation enumeration. 
Our goal was not to use it as a production algorithm, but to test whether real conflict graphs contain enough structural redundancy for such compression.

Our empirical results confirm that this compression can remove many redundant permutations, often reducing $|CG|!$ to a much smaller number of representative orderings. 
However, the remaining number of representatives can still exceed what our implementation can handle efficiently for large conflict groups. 
Thus, while the compression is useful, it does not by itself make exhaustive search practical. 
In these cases, the algorithm must still rely on sampling-based search and evaluate permutations of only a randomly selected subset (L\ref{line:select}). 
This suggests that equivalence-class reduction alone is insufficient as a general optimization under practical resource constraints. 
Further improvements likely need to exploit more specific structure in the conflict graph, and use this structure to design more efficient search algorithms.

\section{Reasoning behind Conjecture}\label{sec:conjecture-intuition}
We conjecture that it is impossible to design a mechanism that simultaneously satisfies \textit{Efficiency} (always selecting the outcome that maximizes social welfare), \textit{Searcher-DSIC}, \textit{Builder-DSIC}, \textit{IR}, and \textit{Budget Balance}.

To understand why, suppose that Builder-DSIC is already achieved (which we will return to shortly). Then, to demand Efficiency and Searcher-DSIC (along with searcher IR), the mechanism must select the block with the highest bid $\beta^*$ and effectively charge each searcher the VCG payment. Equivalently, this means that the winning builder first charges each transaction $i$ its claimed bid $b_i$, and the mechanism subsequently refunds $i$ its marginal contribution to the selected block, defined as $r_i = \beta^* - H_i$, where $H_i = \max_{k \in [0:n]} \set{\text{Welfare (Bid) of } \mathcal{A}_k \text{ without } i}$ denotes the best alternative outcome without transaction $i$. %

Conversely, to incentivize builders to truthfully report their block values (i.e., Builder-DSIC), the mechanism must effectively charge the winning builder the second-highest bid. That is, the winner pays their claimed bid and then is refunded the surplus between their bid and the second-highest one. 

Therefore, the remaining question is whether these two components can coexist in a single mechanism. Unfortunately, in some scenarios, the total refunds to searchers may exceed the revenue collected from the builder, causing the mechanism to subsidize participants and thereby violating the Budget Balance property. Below is a concrete counterexample. 

\begin{example}

Consider two transactions: $\text{tx}_1$ with bid $b_1 = 100$ and $\text{tx}_2$ with bid $b_2 = 1$. There are two algorithms: 
\begin{itemize}
    \item \textit{Algorithm~1:} Always includes the transaction with the smallest hash, assumed to be $\text{tx}_1$.
    \item \textit{Algorithm~2:} Always includes the transaction with the largest hash, assumed to be $\text{tx}_2$.
\end{itemize}
In this case, Algorithm~1 wins with bid $\beta^* = 100$, producing outcome $o^* = (\text{tx}_1)$, while Algorithm~2 is the runner-up with $\beta' = 1$. To satisfy Builder-DSIC, the mechanism effectively charges the winner the second-highest bid, which is $1$. To satisfy Searcher-DSIC, the mechanism computes refunds as follows. Without $\text{tx}_1$, Algorithm~1 yields bid $0$, making Algorithm~2 (bid $1$) the best alternative ($H_1 = 1$), thus the refund to $\text{tx}_1$ is $r_1 = 100 - 1 = 99$. Similarly, $r_2 = 0$. As a result, the mechanism collects $1$ from the builder but pays out $99$ to searchers, resulting in a net deficit.

\end{example}

\input{sections/adoption}

%% file: sections/adoption.tex
\section{Proposer Adoption}\label{sec:adoption}
In this section, we briefly discuss the incentive of proposers, namely, whether they would like to adopt our design, although this is not the focus of this paper. %

\parhead{Adoption game}
We partially answer the question by considering the following setting. The baseline algorithm is the only one in \name, so the strategic players only include the proposer and the searchers. 
We model their interaction as a two-step Stackelberg game, where the proposer, as the leader (namely, moves first), chooses between two strategies:
\begin{enumerate}
  \item[(a)] \emph{Commit:} accept only the outcome produced by our mechanism (the block and payment to the proposer, where refunds to searchers have already been deducted);
  \item[(b)] \emph{Build-and-choose:} build blocks themselves by including all received bundles (with no refunds); and after seeing the outcome of \name, choose whichever block pays the proposer more.
\end{enumerate}

The searchers simultaneously respond as the follower: each can choose to submit their bundle either exclusively to \name or exclusively to the proposer.

The question is: under equilibrium, what strategies will the proposer and searchers adopt? We show that in the following two scenarios, the proposer will choose \textit{Commit} and the searchers will exclusively send their bundles to \name. 
\begin{itemize}
  \item \textit{No conflict}: all bundles are independent, i.e., the execution of one bundle does not affect the bid of another;
  \item \textit{Full conflict}: all bundles are mutually exclusive (e.g., they are targeting the same MEV opportunity) and at most one can be effectively included. %
\end{itemize}

This demonstrates that in certain natural scenarios, the proposer is indeed incentivized to adopt our design. How to incentivize the proposer to adopt a new design in a more general case is an intriguing direction for future work.

\begin{proposition}[Self-enforcing adoption in extreme conflict structures]\label{prop:adoption}
    Consider the strategy profile in which the proposer commits to adopt \name, and every searcher sends their bundle exclusively to \name. If either (i) there is no conflict among bundles, or (ii) there is full conflict among all bundles, then this profile forms a Stackelberg equilibrium.
\end{proposition}

\begin{proof}

    As the baseline block is the only one in \name, a searcher who sends their bundle to \name should always bid truthfully by~\autoref{theorem:searcherDSIC}. Then, whenever the proposer commits to only accept the outcome of our mechanism, a searcher's dominant strategy is to exclusively send their bundle to \name, and obtain non-negative utility according to~\autoref{cor:searchersIR} (otherwise, if the searcher exclusively sends their bundle to the proposer, it is impossible for their bundle to be included in a block, yielding zero utility).

    We analyze the two cases separately, and show that in either case, the proposer who moves first has no incentive to deviate from committing to adopt \name. 
    
    \parhead{Case 1: no conflict} 
    Recall that the proposer has two strategies. 
    If the proposer chooses \textit{Commit}, then as discussed above, all searchers will exclusively send their bundles to \name. By ~\autoref{alg:default}, all bundles will be included in the baseline block, and the bid of each bundle will be fully refunded according to the refund rule in step (2) of \mechanism. As a result, the proposer gains zero utility.  
    If the proposer chooses \textit{Build-and-choose}, then each searcher's best response is still to exclusively send their bundles to \name, yielding the same zero utility for the proposer. Thus, the proposer has no incentive to deviate from choosing \textit{Commit}.

    \parhead{Case 2: full conflict}
    If the proposer chooses \textit{Commit}, then all searchers will exclusively send their bundles to \name. Let $M$ denote the set of all bundles of searchers. Under the structure of full conflict, at most one of them could be effectively included. Our baseline algorithm~\autoref{alg:default} will select the bundle $i$ with the highest bid $b_i = \max_{k\in M} \{b_k\}$. Recall that all searchers bid truthfully as discussed above. Then, \mechanism eventually charges $v_i$ to the searcher $i$, of which the second-highest bid/value $v_{2nd}^M = \max \{ v_k \mid k\in M, k\neq i \}$ goes to the proposer. Thus, the proposer's utility equals the \textit{overall} second-highest value $v_{2nd}^M$. 
    
    On the other hand, if the proposer chooses \textit{Build-and-choose}, suppose that in equilibrium, a subset of bundles $S \subseteq M$ is exclusively sent to the proposer, while others, $M\setminus S$, go to \name. We will show that for any possible $S$, the proposer's utility will be at most $v_{2nd}$. For each bundle in $M\setminus S$, they should remain bidding truthfully. Let $v_{2nd}^{M\setminus S}$ be the second-highest true valuation of bundles in $M\setminus S$. Then the searchers in $S$ essentially face a first-price auction with reserve $v_{2nd}^{M\setminus S}$. Letting $v_{2nd}^S$ be the second-highest true valuation of bundles in $S$, an equilibrium among $S$ is that the searcher with the highest valuation bids $\max\left\{v_{2nd}^{S},v_{2nd}^{M\setminus S}\right\}$ and others bid truthfully. Thus, overall, the proposer's utility is $\max\left\{v_{2nd}^{S},v_{2nd}^{M\setminus S}\right\}\leq v_{2nd}^M$.

    This shows that the maximal utility the proposer can obtain is at most $v_{2nd}^M$, and they can obtain such an amount of money by choosing \textit{Commit}, after which each searcher's dominant strategy is joining \name.

    This concludes the proof.
\end{proof}